\documentclass[envcountsame,envcountsect]{llncs}

\usepackage{authblk}
\usepackage{lineno}
\linenumbers
\usepackage[utf8]{inputenc}
\usepackage{amssymb}
\usepackage{dsfont}
\usepackage{color}

\usepackage{amsmath}
\usepackage{amsfonts}
\usepackage{amssymb}
\usepackage{bbm}
\usepackage{relsize}
\usepackage{amsmath}
\usepackage{graphicx}
\graphicspath{ {./images/} }
\usepackage{colortbl,booktabs,array,tabularx}
\nolinenumbers

\newtheorem{myclaim}[lemma]{Claim}

\newcommand{\negA}{\vspace{-0.05in}}

\newcommand{\negC}{\vspace{-0.18in}}

\newcommand{\mysection}[1]{\negC\section{#1}\negA}

\newcommand{\ceil}[1]{{\left\lceil#1  \right\rceil}}
\newcommand{\comment}[1]{}

\usepackage{microtype}
\usepackage{mathtools,amssymb,mathrsfs}
\usepackage[table,svgnames]{xcolor}
\usepackage[colorlinks=true,linkcolor=black,citecolor=MidnightBlue,urlcolor=MidnightBlue]{hyperref}
\usepackage{xspace}
\usepackage{tikz}
\usepackage{pifont}
\usepackage{adjustbox}
\usepackage{subcaption}
\usepackage{algorithm}
\usepackage{algorithmicx}
\usepackage{graphicx}
\usepackage{algpseudocode}
\usepackage{amsmath}
\DeclareMathOperator*{\argmax}{arg\,max}
\DeclareMathOperator*{\argmin}{arg\,min}

\newcommand{\cm}{{\mathcal{M}}}
\newcommand{\ci}{{\mathcal{I}}}
\newcommand{\cA}{{\mathcal{A}}}
\newcommand{\cB}{{\mathcal{B}}}
\newcommand{\cj}{{\mathcal{J}}}

\newcommand{\cI}{{\mathcal{I}}}
\newcommand{\cC}{{\mathcal{C}}}
\newcommand{\cf}{{\mathcal{F}}}

\newcommand{\OPT}{\textnormal{OPT}}

\newcommand{\eps}{{\varepsilon}}

\newcommand{\floor}[1]{\left\lfloor #1 \right\rfloor}

\newcommand{\Generic}{{Color\_Sets}}

\newtheorem{thm}{Theorem}[section]

\newtheorem{observation}[thm]{Observation}

\def \II   {{\mathcal I}}

\numberwithin{claimcounter}{thm}

\begin{document}
\sloppy
\setcounter{page}{1} 
\title{Approximating Bin Packing with Conflict Graphs via Maximization Techniques}
\author{Ilan Doron-Arad \and Hadas Shachnai}
\institute{Computer Science Department, Technion, Haifa 3200003, Israel. \mbox{E-mail: {\tt \{idoron-arad,hadas\}@cs.technion.ac.il.}}}

\maketitle

\negA
\negA
\negA
\negA

\setcounter{tocdepth}{2}

\begin{abstract}
	We give a comprehensive study of {\em bin packing with conflicts} (BPC).
	The input is a set $I$ of items, sizes $s:I \rightarrow [0,1]$, and a conflict graph $G = (I,E)$. The goal is to find a partition of $I$ into a minimum number of independent sets, each of total size at most $1$. Being a generalization of the notoriously hard graph coloring problem, BPC  
	has been studied mostly on polynomially colorable conflict graphs. 
	An intriguing open question is whether BPC on such graphs admits the same 
	best known approximation guarantees as classic bin packing. 
	
	We answer this question negatively, by showing that (in contrast to bin packing) there is no asymptotic polynomial-time approximation scheme (APTAS) 
	for BPC already on seemingly easy graph classes, such as {\em bipartite} and {\em split} graphs.	We complement this 
result with improved approximation guarantees for BPC on several prominent graph classes. Most notably, we derive 
	 an asymptotic $1.391$-approximation for bipartite graphs,
	 a $2.445$-approximation for perfect graphs, and a $\left(1+\frac{2}{e}\right)$-approximation for split graphs.
	To this end, we introduce a generic framework relying on a novel interpretation of BPC allowing us to solve the problem via {\em maximization} techniques.
	Our framework may find use in tackling BPC on other graph classes arising in applications.
\end{abstract}

\mysection{Introduction}
\label{sec:intro}

We study the {\em bin packing with conflicts (BPC)} problem. We are given a set $I$ of $n$ items, sizes $s:I \rightarrow [0,1]$, and a conflict graph $G = (I,E)$ on the items. A {\em packing} is a partition $(A_1, \ldots,A_t)$ of $I$ into independent sets called {\em bins}, such that for all $b \in \{1,\ldots, t\}$ it holds that $s\left( A_b \right) = \sum_{\ell \in A_b} s(\ell) \leq 1$. The goal is to find a packing in a minimum number of bins. Let ${\ci}=(I,s,E)$
denote a BPC instance. We note that BPC is a generalization of {\em bin packing (BP)} (where $E = \emptyset$) as well as the graph coloring problem (where $s(\ell) = 0~\forall \ell \in I$).\footnote{See the formal definitions of {\em graph coloring} and {\em independent sets} in Section~\ref{sec:preliminaries}.} BPC captures many real-world scenarios such as resource clustering
in parallel computing \cite{beaumont2008distributed}, examination scheduling \cite{laporte1984examination}, database storage \cite{jansen1999approximation}, and product delivery \cite{christofides1979vehicle}. As the special case of 
graph coloring cannot be approximated within a ratio better than $n^{1-\eps}$ \cite{zuckerman2006linear}, most of the research work on BPC has focused on families of conflict graphs which can be optimally colored in polynomial time~\cite{OS95,JO97,jansen1999approximation,mccloskey2005approaches,epstein2008bin,DKS21,doron2022bin,huang2023improved}.

 Let $\OPT=\OPT(\II)$ be the value of an optimal solution for an instance~ $\II$ of a minimization problem~$\mathcal{P}$. 
As in the bin packing problem, we distinguish between {\em absolute} and {\em asymptotic} approximation.
For $\alpha \geq 1$, we say that $\cA$ is an absolute $\alpha$-approximation algorithm for 
$\mathcal{P}$ if for any instance $\II$ of~$\mathcal{P}$ we have $ \cA (\II)/\OPT(\II) \leq \alpha$, where $\cA(\II)$ is the value of the solution returned by $\cA$. Algorithm
$\cA$  is an {\em asymptotic} $\alpha$-approximation algorithm for 
$\mathcal{P}$ if  for any instance $\II$ it holds that $\cA (\II) \leq \alpha \OPT(\II) +o(\OPT(\II))$. 
An APTAS
is a family of algorithms $\{ \cA_{\eps} \}$
such that, for every $\eps>0$, $\cA_{\eps}$ is a polynomial time asymptotic $(1+\eps)$-approximation algorithm for $\mathcal{P}$.  An {\em asymptotic fully polynomial-time
approximation scheme (AFPTAS)} is an APTAS $\{\cA_{\eps} \}$ 
such that $\cA_{\eps} (\II)$ runs in time $\textnormal{poly}(|\II|, \frac{1}{\eps})$, where $|\II|$ is the encoding length of the instance $\II$.

It is well known that, unless P=NP, BP cannot be approximated within ratio better than $\frac{3}{2}$~\cite{garey1979computers}. This ratio is achieved by 
First-Fit Decreasing (FFD)~\cite{S94}.\footnote{We give a detailed description of Algorithm FFD in Appendix~\ref{sec:proofsPrel}.} Also, BP admits an AFPTAS~\cite{karmarkar1982efficient}, and an additive approximation
algorithm which packs any instance ${\ci}$ in at most $\OPT(\II)+O(\log(\OPT(\II)))$ bins~\cite{hoberg2017logarithmic}. 
Despite the wide interest in BPC on polynomially colorable graphs,
the intriguing question whether BPC on such 
graphs admits the same best known approximation guarantees as classic bin packing remained open. 

\negA

\begin{table} [h!]
	\label{table:1}
	\centering
	\begin{tabular}{ c c c c c}
		\hline 
		& \multicolumn{2}{c}{\bf Absolute} & \multicolumn{2}{c}{\bf Asymptotic}\\ 
		\cmidrule(rl){2-3}
		\cmidrule(rl){4-5}
		&  Lower Bound & Upper Bound &  Lower Bound & Upper Bound \\     \hline
		\\  [-2ex]  
		General graphs & $n^{1-\eps}$~\cite{zuckerman2006linear} & $O\left(   \frac{ n (\log \log n)^2 }{(\log n)^3}\right)$~\cite{halldorsson1993still} & $n^{1-\eps}$~\cite{zuckerman2006linear}  &  $O\left(   \frac{ n (\log \log n)^2 }{(\log n)^3}\right)$~\cite{halldorsson1993still} \\  [1ex]  
		Perfect graphs  &  $\mathbf{\cdot}$ &  $\mathbf{2.445}$ ($2.5$ \cite{epstein2008bin}) & $\mathbf{c>1}$ &   $\mathbf{2.445}$ ($2.5$ \cite{epstein2008bin})\\ 	 [1ex]  
		Chordal  graphs&  $\mathbf{\cdot}$ &  $\frac{7}{3}$ \cite{epstein2008bin} & $\mathbf{c>1}$ &  $\frac{7}{3}$ \cite{epstein2008bin} \\  [1ex]   
		Cluster graphs&  $\mathbf{\cdot}$ &  $2$ \cite{U+18} &  & $1$ \cite{doron2022bin}   \\  [1ex]   
		Cluster complement &  $\mathbf{\cdot}$ &  $\mathbf{3/2}$  & $\mathbf{3/2}$ & $\mathbf{3/2}$ \\  [1ex]  
		Split  graphs&  $\mathbf{\cdot}$ &  $\mathbf{1+2/e}$ ($2$ \cite{huang2023improved})  & $\mathbf{c>1}$ & $\mathbf{1+2/e}$ ($2$ \cite{huang2023improved}) \\  [1ex]   
		Bipartite  graphs&  $\mathbf{\cdot}$ &  $\frac{5}{3}$ \cite{huang2023improved} & $\mathbf{c>1}$ & $\mathbf{1.391}$ ($\frac{5}{3}$ \cite{huang2023improved}) \\  [1ex]    
		Partial $k$-trees &  $\mathbf{\cdot}$ &  $2+\eps$ \cite{JO97}  &  & $1$ \cite{jansen1999approximation}  \\  [1ex]    
		Trees & $\mathbf{\cdot}$  &  $\frac{5}{3}$~\cite{huang2023improved} &   & $\mathbf{\cdot}$ \\  [1ex]   
		No conflicts & $\frac{3}{2}$~\cite{garey1979computers}   &  $\frac{3}{2}$~ \cite{simchi1994new}  &   & $1$~\cite{Ro13}\\  [1ex]   
		\hline
	\end{tabular}
	\caption{Known results for Bin Packing with Conflict Graphs}
	\hfill \break
	\label{table:2}
	\centering
\end{table}
\negC
\negC
We answer this question negatively, by showing that (in contrast to bin packing) there is no 
APTAS for BPC even
on seemingly easy graph classes, such as {\em bipartite} and {\em split} graphs.
We complement this
result with improved approximation guarantees for BPC on several prominent graph classes. For BPC on bipartite graphs, we obtain an asymptotic $1.391$-approximation. We further derive improved bounds of $2.445$
for perfect graphs, $\left(1+\frac{2}{e}\right)$ for split graphs, and
$\frac{5}{3}$ for bipartite graphs.\footnote{Recently, Huang et al. \cite{huang2023improved} obtained a $\frac{5}{3}$-approximation for bipartite graphs, simultaneously and independently of our work. We note that the techniques of~\cite{huang2023improved} are different than ours, and their algorithm is more efficient in terms of running time.} Finally, we obtain a tight $\frac{3}{2}$-asymptotic lower bound and an absolute $\frac{3}{2}$-upper bound for graphs that are the complements of cluster graphs (we call these graphs below {\em complete multi-partite}).

Table~\ref{table:2} summarizes the known results for BPC on various classes of 
graphs. New bounds given in this paper are shown in boldface. Entries that are marked with $\boldsymbol{\cdot}$ follow by inference, either by using containment of graph classes (trees are partial $k$-trees), or since the hardness of BPC on all considered graph classes
follows from the hardness of classic BP. Empty entries for lower bounds follow from tight upper bounds.
We give a detailed overview of previous results in Appendix~\ref{sec:related}.

 \noindent {\bf Techniques:} There are several known approaches
 for tackling BPC instances. One celebrated technique  
 introduced by Jansen and {\"{O}}hring~\cite{JO97} relies on finding initially 
 a minimum coloring of the given conflict graph, and then packing
 each color class using a bin packing heuristic, such as First-Fit Decreasing. A notable generalization of this approach is the sophisticated integration of {\em precoloring extension} \cite{JO97,epstein2008bin}, which completes an initial partial coloring of the conflict graph, with no increase to the number of color classes. Another elegant technique is a matching-based algorithm, applied by Epstein and Levin \cite{epstein2008bin} and by Huang et al. \cite{huang2023improved}. 
  
 The best known algorithms (prior to this work), e.g., for
  perfect graphs~\cite{epstein2008bin} and split graphs~\cite{huang2023improved}  are based on the above techniques. While the analyses of these algorithms are tight, the approximation guarantees do not match the existing
  lower bounds for BPC on these graph classes; 
  thus, obtaining improved approximations requires new techniques.
   
  In this paper we present a novel point of view of BPC involving 
  the solution of a maximization problem as a subroutine.
  We first find an {\em initial packing} of a subset $S \subseteq I$ of items, which serves as a baseline packing with {\em high potential} for adding items (from $I \setminus S$) without increasing the number of bins used. 
The remaining items are then 
assigned to extra bins using a simple heuristic. 
Thus, given a BPC instance, our framework consists of the following main steps.
\negA
  \begin{enumerate}
  	\item
  	Find an initial packing $\cA = (A_1, \ldots, A_m)$ of high potential for $S \subseteq I$. 
   
  	\item Maximize the total size of items in $\cA$ by adding items in $I \setminus S$. 
  	\item Assign the remaining (unpacked) items to extra bins using a greedy approach respecting the conflict graph constraints. 
  \end{enumerate} 
\negA 
The above generic framework reduces BPC to cleverly finding an initial packing of high potential,  and then efficiently approximating the corresponding maximization problem, while exploiting structural properties of the given conflict graph. One may view classic approaches for solving BP (e.g., \cite{fernandez1981bin}), as an application of this technique: 
find an initial packing of high potential containing the {\em large} items;
then add the {\em small} items using First-Fit. In this setting, the tricky part is to find an initial high potential packing, while adding the small items is trivial. However, in the presence of a conflict graph, solving the induced maximization problem is much more challenging.

Interestingly, we are able to obtain initial packings of high potential for BPC on several conflict graph classes. 
 To solve the maximization problem, we first
 derive efficient approximation for maximizing the total size of items within a {\em single} bin. Our algorithm is based on finding a maximum weight independent set of {\em bounded} total size in the graph, 
 combined with enumeration over items of large sizes. 
 Using the single bin algorithm, the maximization problem is solved via
 application of the {\em separable assignment problem (SAP)} \cite{fleischer2011tight} framework, adapted to our setting. 
 Combined with a hybrid of several techniques (to efficiently handle different types of instances) this leads to
improved bounds for BPC on perfect, split, and bipartite graphs (see Sections~\ref{sec:BP}, \ref{sec:split}, and Appendix~\ref{sec:53}). 
 Our framework may find use in tackling BPC on other graph classes arising in applications. 

\noindent{\bfseries Organization:} In section~\ref{sec:preliminaries} we give some definitions and preliminary results. Section~\ref{sec:BP} presents an approximation algorithm for BPC on perfect graphs and an asymptotic approximation on bipartite graphs. In Section~\ref{sec:split} we give an algorithm for split graphs. We present our hardness results in Section~\ref{sec:HardnessBP} and conclude in Section~\ref{sec:discussion}. Due to space constraints, some of our results are deferred to the Appendix. For convenience, the last page of the paper includes a table of contents. 
\mysection{Preliminaries}
\label{sec:preliminaries}

For any $k \in \mathbb{R}$, let $[k] = \{1,2,\ldots,\floor{k}\}$. 
Also, for a function $f:A \rightarrow \mathbb{R}_{\geq 0}$ and a subset of elements $C \subseteq A$, we define $f(C) = \sum_{e \in C} f(e)$.

  \noindent {\bf Coloring and Independent Sets:} Given a graph $G = (V,E)$, an {\em independent set} in $G$ is a subset of vertices $S \subseteq V$ such that for all $u,v \in S$ it holds that $(u,v) \notin E$. Let $\textsf{IS}(G)$ be the collection of all independent sets in $G$. Given weight function $w:V \rightarrow \mathbb{R}_{\geq 0}$, a {\em maximum independent set w.r.t.} $w$ is an independent set $S \in \textnormal{\textsf{IS}}(G)$ such that $w(S)$ is maximal. A {\em coloring} of $G$ is a partition $(V_1, \ldots, V_t)$ of $V$ such that $\forall i \in [t]: V_i \in \textsf{IS}(G)$; we call each subset of vertices $V_i$ {\em color class} $i$. Let $\chi(G)$ be the minimum number of colors required for a coloring of $G$.
 A graph $G$ is {\em perfect} if for every induced subgraph $G'$ of $G$ the cardinality of the maximal clique of $G'$ is equal to 
 $\chi(G')$; note that $G'$ is also a perfect graph. 
 The following well known result is due to \cite{grotschel2012geometric}. 	\begin{lemma}
 	\label{lem:grot}
 	Given a perfect graph $G = (V,E)$, a minimum coloring of $G$ and a maximum weight independent set of $G$ can be computed in polynomial time. 
 \end{lemma}

  \noindent {\bf Bin Packing with Conflicts:} Given a BPC instance $\ci$, let $G_{\ci} = (I,E)$ denote the conflict graph of $\ci$. 
A {\em packing} of a subset of items $S \subseteq I$
is a partition $\cB = (B_1, \ldots,B_t)$ of $S$ such that, for all $i \in [t]$, $B_i$ is an independent set in $G_{\ci}$, and $s(B_i) \leq 1$. Let $\#\cB$ be the number of bins (i.e., entries) in $\cB$. 

In this paper we consider BPC on several well studied classes of perfect graphs and the acronym BPC refers from now on to perfect conflict graphs.     
 For {\em bin packing with bipartite conflicts (BPB)}, where the conflict graph is bipartite,
we assume a bipartition of $V$ is known and given by $X_{V}$ and $Y_V$.
Recall that $G = (V,E)$ is a split graph if there is a partition $K,S$ of $V$ into a clique and an independent set, respectively. We call this variant of BPC {\em bin packing with split graph conflicts (BPS)}.

The following notation will be useful while enhancing a partial packing by new items. For two packings $\cB = (B_1, \ldots,B_t)$ and $\cC = (C_1, \ldots,C_r)$, let $\cB\oplus \cC = (B_1, \ldots, B_t, C_1, \ldots, C_r)$ be the {\em concatenation} of $\cB$ and $ \cC$; also, for $t = r$ let $\cB+\cC = (B_1 \cup C_1, \ldots, B_t \cup C_t)$ be the {\em union} of the two packings; note that the latter is not necessarily a packing. We denote by $\textsf{items}(\cB) = \bigcup_{i \in [t]} B_i$ the set of items in the packing $\cB$. Finally, let $\ci = (I,s,E)$ be a BPC instance and $T \subseteq I$ a subset of items. 
 Define the BPC instances $\ci \cap T = (T,s, E_T)$ and $\ci \setminus T = (I \setminus T,s,E_{I \setminus T})$ where for all  $X \in \{T, I \setminus T\}$ $E_X = \{(u,v) \in E~|~u,v \in X\}$.

 \noindent {\bf Bin Packing Algorithms:} 
We use $\ci = (I,s)$ to denote a BP instance, where $I$ is a set of $n$ items 
for some $n \geq 1$, and $s:I \rightarrow [0,1]$ is the size function. Let $L_{\ci} = \{\ell \in I~|~s(\ell)>\frac{1}{2}\}$ be the set of {\em large} items, $M_{\ci} = \{\ell \in I~|~\frac{1}{3}<s(\ell)\leq\frac{1}{2}\}$ the set of {\em medium} items, and $S_{\ci} = \{\ell \in I~|~s(\ell)\leq\frac{1}{3}\}$ the set of {\em small} items. Our algorithms use as building blocks also algorithms for BP. The results in the next two lemmas are tailored for our purposes. We give the detailed proofs in Appendix~\ref{sec:proofsPrel}.\footnote{For more details on algorithms {\sf FFD} and {\sf AsymptoticBP} see, e.g., \cite{vazirani}.}

\begin{lemma}
	\label{lem:FF}
	Given a \textnormal{BP} instance $\ci = (I,s)$,
	there is a polynomial-time algorithm \textnormal{\textsf{First-Fit Decreasing} (FFD)} which  returns a packing $\cB = (B_1, \ldots,B_t)$ of $\ci$ where  $
	\#\cB \leq (1+2 \cdot \max_{\ell \in I} s(\ell) )\cdot s(I)+1$. Moreover, it also holds that $
	\#\cB \leq |L_{\ci}|+\frac{3}{2} \cdot s(M_{\ci})+\frac{4}{3} \cdot s(S_{\ci})+1$.
\end{lemma}

\begin{lemma}
	\label{lem:rothvos}
	Given a \textnormal{BP} instance $\ci = (I,s)$,
	there is a polynomial-time algorithm \textnormal{\textsf{AsymptoticBP}} which returns a packing $\cB = (B_1, \ldots,B_t)$ of $\ci$ such that $t = \OPT(\ci)+o(\OPT(\ci))$. Moreover, if $\OPT(\ci) \geq 100$ then $t \leq 1.02 \cdot \OPT(\ci)$. 
\end{lemma}

\mysection{Approximations for Perfect and Bipartite Graphs}
\label{sec:BP}
In this section we consider the 
bin packing problem with a perfect or bipartite conflict graph. Previous works (e.g., \cite{JO97}, \cite{epstein2008bin}) 
showed the usefulness of 
the approach based on finding first a minimal coloring of the given conflict graph, and then packing each color class as a separate
bin packing instance (using, e.g., algorithm \textsf{\textsf{FFD}}). Indeed, this approach yields efficient approximations for BPC;
however, it does reach a certain limit. To enhance the performance of this {\em coloring based} approach, we design several subroutines. Combined, they cover the problematic cases and lead to improved approximation guarantees (see Table~\ref{table:2}). 

Our first subroutine is the coloring based approach, with a simple modification to improve the asymptotic performance. For each color class $C_i, i = 1, \ldots,k$ in a minimal coloring of the given conflict graph, we find a packing of $C_i$ using \textsf{\textsf{FFD}}, and another packing using \textsf{AsymptoticBP} (see Lemma~\ref{lem:rothvos}). We choose the packing which has smaller number of bins. Finally, the returned packing is the concatenation of the packings of all color classes. The pseudocode of Algorithm  
\textsf{\Generic} is given in Algorithm~\ref{alg:generic}.

\begin{algorithm}[h]
		\caption{$\textsf{\Generic}(\ci =(I,s,E))$}
		\label{alg:generic}
		\begin{algorithmic}[1]
			\State{Compute a minimal coloring $\cC = (C_1, \ldots,C_k)$ of $G_{\cI}$.\label{step:GenericCol}}
			\State{Initialize an empty packing $\cB \leftarrow ()$.\label{step:GenericEmpty}}
			\For {$i \in [k]$\label{step:GenericFor}}
			\State{Compute $\cA_1 \leftarrow \textsf{\textsf{FFD}}((C_i,s))$ and $\cA_2 \leftarrow \textsf{AsymptoticBP}((C_i,s))$.\label{step:GenericA}}
			\State{$\cB \leftarrow \cB \oplus \argmin_{\cA \in \{\cA_1, \cA_2\}} \#\cA$.\label{step:GenericB}}
			\EndFor
			\State{Return $\cB$.\label{step:GenericReturn}}
		\end{algorithmic}
\end{algorithm}

For the remainder of this section, fix a BPC instance $\ci = (I,s,E)$. The performance guarantees of Algorithm \textsf{\Generic} are stated in the next lemma. 
\begin{lemma}
	\label{lem:Generic}
		\label{claim:gen2}
	Given a \textnormal{BPC} instance $\ci = (I,s,E)$, Algorithm \textnormal{\textsf{\Generic}} returns in polynomial time in $|\ci|$ a packing $\cB$ of $\ci$ such that $\#\cB \leq \chi(G_{\ci})+|L_{\ci}|+\frac{3}{2} \cdot s(M_{\ci})+\frac{4}{3} \cdot s(S_{\ci})$. Moreover, if $\ci$ is a \textnormal{BPB} instance then $\#\cB \leq \frac{3}{2} \cdot |L_{\ci}|+\frac{4}{3} \cdot \left(\OPT(\ci)-|L_{\ci}|\right)+o(\OPT(\ci))$. 
\end{lemma}

Note that the bounds may not be tight for instances with many large items. Specifically, if $|L_{\ci}| \approx \OPT(\ci)$ then a variant 
of Algorithm \textsf{\Generic} was shown to yield a packing of at least $2.5 \cdot \OPT(\ci)$ bins \cite{epstein2008bin}. 
To overcome this, we use an approach based on the simple yet crucial observation that there can be at most one large item in a bin. Therefore, we view the large items as {\em bins} and assign items to these bins to maximize the total size packed in bins including large items. We formalize the problem initially on a single bin. 

\begin{definition}
	\label{def:boundedIS}
	In the {\em bounded independent set problem (BIS)} we are given a graph $G = (V,E)$, a weight function $w:V \rightarrow \mathbb{R}_{\geq 0}$, and a budget $\beta \in \mathbb{R}_{\geq 0}$. The goal is to find an independent set $S \subseteq V$ in $G$ such that $w(S)$ is maximized and $w(S) \leq \beta$. Let $\cI = (V,E,w,\beta)$ be a BIS instance.
\end{definition}

Towards solving BIS, we need the following definitions.
For $\alpha \in (0, 1]$, $\cA$ is an $\alpha$-approximation algorithm for a maximization problem $\mathcal{P}$ if, for any instance $\cI$ of $\mathcal{P}$,
	$\cA$ outputs a solution of value at least $\alpha \cdot OPT(\cI)$. A 
	{\em polynomial-time approximation scheme (PTAS)} for  $\mathcal{P}$ 
		is a family of algorithms $\{ A_{\eps} \}$
	such that, for any $\eps>0$, $A_{\eps}$ is a
	polynomial-time $(1 - \eps)$-approximation algorithm for $\mathcal{P}$. 
	A {\em fully PTAS (FPTAS)} is a PTAS $\{ A_{\eps} \}$ where, for all $\eps>0$, $A_{\eps}$ is polynomial also in $\frac{1}{\eps}$.
We now describe a PTAS for BIS.
Fix a BIS instance $\ci = (V,E,w,\beta)$ and $\eps>0$. As there can be at most $\eps^{-1}$ items with weight at least $\eps \cdot \beta$ in some optimal solution $\OPT$ for $\ci$, we can {\em guess} this set $F$ of items via enumeration.  Then, to add smaller items to $F$, we define a residual graph  $G_F$ of items with weights at most $\eps \cdot \beta$ which are not adjacent to any item in $F$. Formally, define $G_F = (V_F,E_F)$, where
$$
	V_F = \{v \in V \setminus F~|~w(v) \leq \eps \cdot \beta, \forall u \in F: (v,u) \notin E\},   ~E_F = \{(u,v) \in E~|~ u,v \in V_F\}
$$

Now, we find a maximum weight independent set $S$ in $G_F$. Note that this can be done
in polynomial time for perfect and bipartite graphs. If $w(F \cup S) \leq \beta$ then we have an optimal solution; otherwise, we discard iteratively items from $S$ until the remaining items form a feasible solution for $\ci$. Since we discard only items with relatively small weights, we lose only an $\eps$-fraction of the weight 
relative to the optimum. The pseudocode for the scheme is given in Algorithm~\ref{alg:ptas}. 
\negA
\begin{algorithm}[h]
	\caption{$\textsf{PTAS}((V,E,w,\beta),\eps)$}
	\label{alg:ptas}
		\begin{algorithmic}[1]
		\State{Initialize $A \leftarrow \emptyset$.\label{step:initptas}}
		\For {all independent sets $F \subseteq V$ in $(V,E)$ $\text{ s.t. } |F| \leq \eps^{-1}, w(F) \leq \beta$\label{step:forptas}}
		\State{Define the residual graph $G_F = (V_F, E_F)$.\label{step:EF}}
		\State{Find a maximum independent set $S$ of $G_F$ w.r.t. $w$.\label{step:ISptas}}
		\While{$w(F \cup S) > \beta$\label{step:whileptas}}
		\State{Choose arbitrary $z \in S$.}
		\State{Update $S \leftarrow S \setminus \{z\}$.\label{step:updateSptas}}
		\EndWhile
		\If{$w(A) < w(F \cup S) \label{step:ifptas}$}
		\State{Update $A \leftarrow F \cup S$.\label{step:Aptas}}
		\EndIf
		\EndFor
		\State{Return $A$.\label{step:returnptas}}
	\end{algorithmic}
\end{algorithm}

\begin{lemma}
	\label{lem:ptas}
	Algorithm~\ref{alg:ptas} is a \textnormal{PTAS} for \textnormal{BIS}. 
\end{lemma}

We now define our maximization problem for multiple bins. We solve a slightly generalized problem in which we have an initial partial packing in $t$ bins. Our goal is to add to these bins (from unpacked items) a subset of items of maximum total size. Formally,

\begin{definition}
	\label{def:maxSAP}
	Given a \textnormal{BPC} instance $\ci = (I,s,E)$, $S \subseteq I$, and a packing $\cB = (B_1, \ldots,B_t)$ of $S$, define the {\em maximization problem} of $\ci$ and $\cB$ as the problem of finding a packing $\cB+\cC$ of $S \cup T$, where $T \subseteq I \setminus S$ and $\cC = (C_1, \ldots,C_t)$ is a packing of $T$, such that $s(T)$ is maximized.
\end{definition}

Our solution for BIS is used to obtain a $(1-\frac{1}{e}-\eps)$-approximation for the maximization problem described in Definition~\ref{def:maxSAP}. This is done using the approach of \cite{fleischer2011tight} for the more general {\em separable assignment problem (SAP)}. 
\begin{lemma}
	\label{lem:SAP}
Given a \textnormal{BPC} instance $\ci = (I,s,E)$, $S \subseteq I$, a packing $\cB = (B_1, \ldots,B_t)$ of $S$, and a constant $\eps>0$, there is an algorithm  \textnormal{\textsf{MaxSize}} which returns in time polynomial in $|\ci|$ a $(1-\frac{1}{e}-\eps)$-approximation for the maximization problem of $\ci$ and $\cB$. Moreover, given an \textnormal{FPTAS} for \textnormal{BIS} on the graph $(I,E)$, the weight function $s$, and the budget $\beta = 1$,
\textnormal{\textsf{MaxSize}}
is a $(1-\frac{1}{e})$-approximation algorithm for the maximization problem of $\ci$ and $\cB$. 
\end{lemma}

We use the above to obtain a feasible solution for the instance. This is done via a reduction to the maximization problem of the instance with a singleton packing of the large items and packing the remaining items in extra bins. Specifically, in the subroutine \textsf{MaxSolve}, we initially put each item in $L_{\ci}$ in a separate bin. Then, additional items from $S_{\ci}$ and $M_{\ci}$ are added to  the bins using Algorithm $\textsf{MaxSize}$. The remaining items are packed using Algorithm \textsf{\Generic}. The pseudocode of the subroutine \textsf{MaxSolve} is given in Algorithm~\ref{alg:MaxSolve}. 

\begin{algorithm}[h]
	\caption{$\textsf{MaxSolve}(\ci =(I,s,E), \eps)$}
	\label{alg:MaxSolve}
		\begin{algorithmic}[1]
		\State{Define $T \leftarrow \left(\{\ell\}~|~\ell \in L_{\ci}\right)$.\label{step:maxSolvedef}}
		\State{$\cA \leftarrow \textsf{MaxSize}(\ci, L_{\ci},  T, \eps)$.\label{step:maxSolveSize}}
		\State{$\cB \leftarrow \textsf{\Generic}(\ci \setminus \textsf{items}(\cA))$.\label{step:maxSolveGeneric}}
		\State{Return $\cA \oplus \cB$.\label{step:maxSolveReturn}}
		\end{algorithmic}
\end{algorithm}
The proof of Lemma~\ref{lem:MaxSolve} uses Lemmas~\ref{lem:Generic}, ~\ref{lem:ptas}, and~\ref{lem:SAP}. 
\begin{lemma}
	\label{lem:MaxSolve}
	Given a \textnormal{BPC} instance $\ci = (I,s,E)$ and an $\eps>0$, Algorithm \textnormal{\textsf{MaxSolve}} returns in polynomial time in $|\ci|$ a packing $\cC$ of $\ci$ such that there are $0 \leq x \leq s(M_{\ci})$ and $0\leq y \leq s(S_{\ci})$ such that the following holds. \begin{enumerate}
		\item $x+y \leq \OPT(\ci)-|L_{\ci}|+\left(\frac{1}{e}+\eps\right)\cdot\frac{|L_{\ci}|}{2}$.
		\item $\#\cC \leq \chi(G_{\ci})+|L_{\ci}|+\frac{3}{2} \cdot x+\frac{4}{3} \cdot y$.
	\end{enumerate}
\end{lemma}

Lemma~\ref{lem:MaxSolve} improves significantly the performance of Algorithm \textsf{\Generic} for instances with many large items. 
However, Algorithm \textsf{MaxSize} may prefer
small over medium items;
the latter items will
be packed by Algorithm \textsf{\Generic} (see Algorithm~\ref{alg:MaxSolve}). The packing of these medium items may harm the approximation guarantee.  
Thus, to tackle instances with many medium items, we use a reduction to a maximum matching problem for packing the large and medium items in at most $\OPT(\ci)$ bins.\footnote{We note that a maximum matching based technique for BPC is used also in \cite{epstein2008bin,huang2023improved}.} Then, the remaining items can be packed using Algorithm \textsf{\Generic}. The graph used for the following subroutine \textsf{Matching} contains all large and medium items;
there is an edge
between any two items which can be assigned to the same bin 
in a packing of the instance ${\ci}$. Formally, \begin{definition}
	\label{def:GraphMatching}
	Given a \textnormal{BPC} instance $\ci = (I,s,E)$, the {\em auxiliary graph of $\ci$} is $H_{\ci} = (L_{\ci} \cup M_{\ci}, E_H)$, where $E_H = \{(u,v)~|~u,v \in L_{\ci} \cup M_{\ci}, s(\{u,v\}) \leq 1, (u,v) \notin E\}$. 
\end{definition} 
Algorithm \textsf{Matching} finds a maximum matching in $H_{\ci}$
and outputs a packing of the large and medium items where pairs of items taken to the matching are packed together, and the remaining items are packed in extra bins using Algorithm \textsf{\Generic}. The pseudocode of the subroutine \textsf{Matching} is given in Algorithm~\ref{alg:Matching}.

\begin{algorithm}[h]
	\caption{$\textsf{Matching}(\ci =(I,s,E))$}
	\label{alg:Matching}
		\begin{algorithmic}[1]
	\State{Find a maximum matching $\cm$ in $H_{\ci}$.\label{step:matchingMax}}
 	\State{$\cB \leftarrow \left(\{u,v\}~|~(u,v) \in \cm\right) \oplus \left(\{v\}~|~v \in M_{\ci} \cup L_{\ci}, \forall u \in M_{\ci} \cup L_{\ci}: (u,v) \notin \cm\right)$.\label{step:matchingB}}
		\State{Return $\cB \oplus \textsf{\Generic}(\ci \setminus (M_{\ci} \cup L_{\ci}))$.\label{step:matchingReturn}}
		\end{algorithmic}
\end{algorithm}
\negA

The proof of Lemma~\ref{lem:Matching} follows by noting that the cardinality of a maximum matching in $H_{\ci}$
in addition to the number of unmatched vertices in $L_{\ci} \cup M_{\ci}$ is at most $\OPT(\ci)$.
\begin{lemma}
	\label{lem:Matching}
	Given a \textnormal{BPC} instance $\ci = (I,s,E)$, Algorithm \textnormal{\textsf{Matching}} returns in polynomial time in $|\ci|$ a packing $\cA$ of $\ci$ such that  $\#\cA \leq \OPT(\ci)+ \chi(G_{\ci})+\frac{4}{3} \cdot s(S_{\ci})$.
\end{lemma}

We now have the required components for the approximation algorithm for BPC and the asymptotic approximation for BPB. 
Our algorithm, \textsf{ApproxBPC}, applies all of the above subroutines and returns the packing which uses the smallest number of bins.
We use $\eps = 0.0001$ for the error parameter in \textnormal{\textsf{MaxSolve}}. The pseudocode of \textsf{ApproxBPC} is given in Algorithm~\ref{alg:BPC}.  
\negA
\begin{algorithm}[h]
	\caption{$\textsf{ApproxBPC}(\ci)$}
	\label{alg:BPC}
		\begin{algorithmic}[1]
		\State{Let $\eps = 0.0001$.\label{step:BPCeps}}
		\State{Compute $\cA_1 \leftarrow \textsf{\Generic}(\ci)$, $\cA_2 \leftarrow \textsf{MaxSolve}(\ci, \eps)$, $\cA_3 \leftarrow \textsf{Matching}(\ci)$.\label{step:BPCA}}
	
		\State{Return $\argmin_{\cA \in \{\cA_1, \cA_2, \cA_3\}} \#\cA$.\label{step:BPCReturn}}
		\end{algorithmic}
\end{algorithm}
\negA

We give below the main result of this section. The proof follows by the argument that the subroutines $ \textsf{\Generic}$, $\textsf{MaxSolve}$, and $ \textsf{Matching}$ handle together most of the difficult cases. Specifically, if the instance contains many large items, then $\textsf{MaxSolve}$ produces the best approximation. If there are many large and medium items, then $ \textsf{Matching}$ improves the approximation guarantee. Finally, for any other case, our analysis of the $ \textsf{\Generic}$ algorithm gives us the desired ratio. We summarize with the next result. 
\begin{theorem}
	\label{thm:BPC}
Algorithm~\ref{alg:BPC} is a $2.445$-approximation for \textnormal{BPC} and an asymptotic $1.391$-approximation for \textnormal{BPB}. 
\end{theorem}

\mysection{Split Graphs}
\label{sec:split}

 In this section we enhance the use of maximization techniques for BPC 
 to obtain an absolute approximation algorithm for BPS. In particular, we improve upon the recent result of Huang et al. \cite{huang2023improved}. We use 
 as a subroutine 
 the maximization technique as outlined in Lemma~\ref{lem:SAP}. Specifically, we start by obtaining an FPTAS for the BIS problem
 on split graphs. For the following, fix a BPS instance $\ci = (I,s,E)$. It is well known (see, e.g., \cite{golumbic2004algorithmic}) that
 a partition of the vertices 
 of a split graph into a clique and an independent set can be found 
 in polynomial time. Thus, for simplicity we assume that such a partition of the split graph $G$ is known and given by $K_{G}, S_{G}$. We note that an FPTAS for the BIS problem on split graphs follows from a result of Pferschy and Schauer~\cite{pferschy2009knapsack} for {\em knapsack with conflicts}, since split graphs are a subclass of chordal graphs. We give a simpler FPTAS for our problem in Appendix~\ref{sec:proofsSplit}. 
 
	\begin{lemma}
	\label{lem:KP2}
	There is an algorithm \textnormal{\textsf{FPTAS-BIS}} that is an \textnormal{FPTAS} for the \textnormal{BIS} problem on split graphs. 
\end{lemma}

Our next goal is to find a suitable initial packing $\cB$ to which we apply \textsf{MaxSize}.
Clearly, the vertices $K_{G_{\ci}}$ must be assigned to different bins. Therefore, our initial packing contains the vertices of $K_{G_{\ci}}$ distributed to $|K_{G_{\ci}}|$ bins as $\{ \{v\}~|~v \in K_{G_{\ci}} \}$. In addition, let $\alpha \in \left\{0,1,\ldots, \ceil{2 \cdot s(I)}+1 \right\}$ be a {\em guess} of $\OPT(\ci)-|K_{G_{\ci}}|$; then, $(\emptyset)_{i \in [\alpha]}$ is a packing of $\alpha$ bins that do not contain items. Together, the two above packings form the initial packing $\cB_{\alpha}$. Our algorithm uses \textsf{MaxSize} to add items to the existing bins of $\cB_{\alpha}$ and packs the remaining items using \textsf{FFD}.
Note that we do not need an error parameter $\eps$, since we use \textsf{MaxSize} with an FPTAS (see Lemma~\ref{lem:SAP}).
For simplicity, we assume that $\OPT(\ci) \geq 2$
(else we can trivially pack the instance in a single bin). 
We give the pseudocode of our algorithm for BPS in Algorithm~\ref{alg:split}.

\negA
\begin{algorithm}[h]
	\caption{$\textsf{Split-Approx}(\ci = (I,s,E))$}
	\label{alg:split}
	\begin{algorithmic}[1]
		\For{$\alpha \in \left\{0,1,\ldots, \ceil{2 \cdot s(I)} +1\right\}$}
	\State{Define $\cB_{\alpha} =  \{ \{v\}~|~v \in K_{G_{\ci}} \} \oplus (\emptyset)_{i \in [\alpha]}$}
	\State{$\cA_{\alpha} \leftarrow \textsf{MaxSize}(\ci, K_{G_\ci}, \cB_{\alpha})$.\label{step:split1}}
	\State{$\cA^*_{\alpha} \leftarrow \cA_{\alpha} \oplus \textsf{FFD}(\ci \setminus \textsf{items}(\cA_{\alpha}))$.\label{step:split3}}
	\EndFor
	\State{Return $\argmin_{ \alpha \in \left\{0,1,\ldots, \ceil{2 \cdot s(I)} +1 \right\}} \#\cA^*_{\alpha}$.\label{step:split2}}
	\end{algorithmic}
\end{algorithm}
\negA

By Lemmas~\ref{lem:KP2} and~\ref{lem:SAP} we have a $\left( 1-\frac{1}{e}\right)$-approximation for the maximization problem of the BPS instance $\cI$ and an initial partial packing $\cB$. Hence, for a correct guess $\alpha = \OPT(\ci)-|K_{G_{\ci}}|$, the remaining items to be packed by \textsf{FFD} are of total size at most $\frac{s(I)}{e}$ and can be packed in $\frac{2 \cdot \OPT(\ci)}{e}$ bins. Thus, we have 
\begin{theorem}
	\label{thm:split}
	Algorithm~\ref{alg:split} is a $\left( 1+ \frac{2}{e} \right)$-approximation for \textnormal{BPS}.
\end{theorem}

\mysection{Asymptotic Hardness for Bipartite and Split Graphs}
\label{sec:HardnessBP}

In this section we show that there is no APTAS for BPB and BPS, unless $P = NP$. We use a reduction from the {\em Bounded 3-dimensional matching (B3DM)} problem, that is known to be MAX SNP-complete~\cite{kann1991maximum}.

For the remainder of this section, let $c>2$ be some constant. 
A B3DM instance is a four-tuple $\cj = (X,Y,Z,T)$, where $X,Y,Z$ are three disjoint finite sets and $T \subseteq X \times Y \times Z$; also, for each $u \in X \cup Y \cup Z$ there are at most $c$ triples in $T$ to which $u$ belongs. A {\em solution} for $\cj$ is $M \subseteq T$ such that for all $u \in X \cup Y \cup Z$ it holds that $u$ appears in at most one triple of $M$.  The objective is to find a solution $M$ of maximal cardinality. Let $\OPT(\cj)$ be the value of an optimal solution for $\cj$. 
We use in our reduction a {\em restricted} instance of B3DM defined as follows.
\begin{definition}
	\label{def:restrictedInstance}
	For $k \in \mathbb{N}$, a \textnormal{B3DM} instance $\cj$ 
	is $k$-{\em restricted} if $\OPT(\cj) \geq k$. 
\end{definition}

In the next lemma we show the hardness of $k$-restricted B3DM. Intuitively, since B3DM instances $\cj$ with $\OPT(\cj) \leq k$ are polynomially solvable for a fixed $k$ (e.g., by exhaustive enumeration), it follows that restricted-B3DM must be hard to approximate, by the hardness result of Kann~\cite{kann1991maximum}. 
\begin{lemma}
	\label{1-rest}
	There is a constant $\alpha>1$ such that for any $k \in \mathbb{N}$ there is no $\alpha$-approximation for the $k$-restricted \textnormal{B3DM} problem unless \textnormal{P=NP}. 
\end{lemma}

We give below the main idea of our reduction, showing the asymptotic hardness of BPB and BPS. A more formal description and the proof of Lemma~\ref{1-rest} are given in Appendix~\ref{sec:proofsBPHardness}. For a sufficiently large $n \in \mathbb{N}$, let $\cj = (X,Y,Z,T)$ be an $n$-restricted instance of B3DM, and let the components of $\cj$, together with appropriate indexing, be 
$U =X \cup Y \cup Z$ and $T$, where
\begin{equation*}
	\begin{aligned}
		X ={} & \{x_1, \ldots, x_{\tilde{x}}\}, Y = \{y_1, \ldots, y_{\tilde{y}}\}, Z = \{z_1, \ldots, z_{\tilde{z}}\}, T = \{t_1, \ldots, t_{\tilde{t}}\}. 
	\end{aligned}
\end{equation*}

We outline our reduction for BPB and later show how it can be modified to yield the hardness result for BPS. Given an $n$-restricted B3DM instance, we construct a sequence of BPB instances. Each BPB instance contains an item for each element $u \in U$, and an item for each triple $t \in T$. There is an edge $(u,t)$ if $u \in U$ and $t \in T$, and $u$ does not appear in $t$, i.e., we forbid packing an element $u$ in the same bin with a triple not containing $u$, for any $u \in U$. 
Since we do not know the exact value of $\OPT(\cj)$,
we define a family of instances with different number of {\em filler items}; these items are packed in the optimum of our constructed BPB instance together with elements not taken to the solution for $\cj$.

Specifically, for a {\em guess} $i \in \{n,n+1, \ldots, |T|\}$ of $\OPT(\cj)$, we define a BPB instance $\mathcal{I}_i = (I_i,s,E)$. The set of items in $\cI_i$ is $I_i = U \cup P_i \cup T \cup Q_i$, where $P_i,Q_i$ are a set of $\tilde{t}-i$ (filler) items and a set of $\tilde{x}+\tilde{y}+\tilde{z}-3 \cdot i$ (filler) items, respectively, such that $P_i \cap U = \emptyset$ and $Q_i \cap U = \emptyset$.
The bipartite (conflict) graph of $\cI_i$ is $G_i = (I_i,E)$, where $E = E_X \cup E_Y \cup E_Z$ is defined as follows.

\begin{equation*}
	\begin{aligned}
		E_X ={} & \{ (x,t)~|~ x \in X, t = (x',y,z) \in T, x \neq x' \}\\
		E_Y ={} & \{ (y,t)~|~ y \in Y, t = (x,y',z) \in T, y \neq y' \}
		\\
		E_Z ={} & \{ (z,t)~|~ z \in Z, t = (x,y,z') \in T, z \neq z' \}
		\\
	\end{aligned}
\end{equation*}

Finally, define the sizes of items in $\cI_i$ to be
$$\forall u \in U, p \in P_i, q \in Q_i, t \in T: ~ s(u) = 0.15, s(p) = 0.45, s(q) = 0.85, s(t) = 0.55.$$
By the above, the only way to pack three items from $x,y,z \in U$ with a triple $t \in T$ is if $(x,y,z) = t$; also, $s\left(\{x,y,z,t\}\right) = 1$. 
For an illustration of the reduction see Figure~\ref{fig:proof}.

	\begin{figure}[htbp]
	\hspace*{+1cm}                                                           
	\includegraphics[scale=0.25]{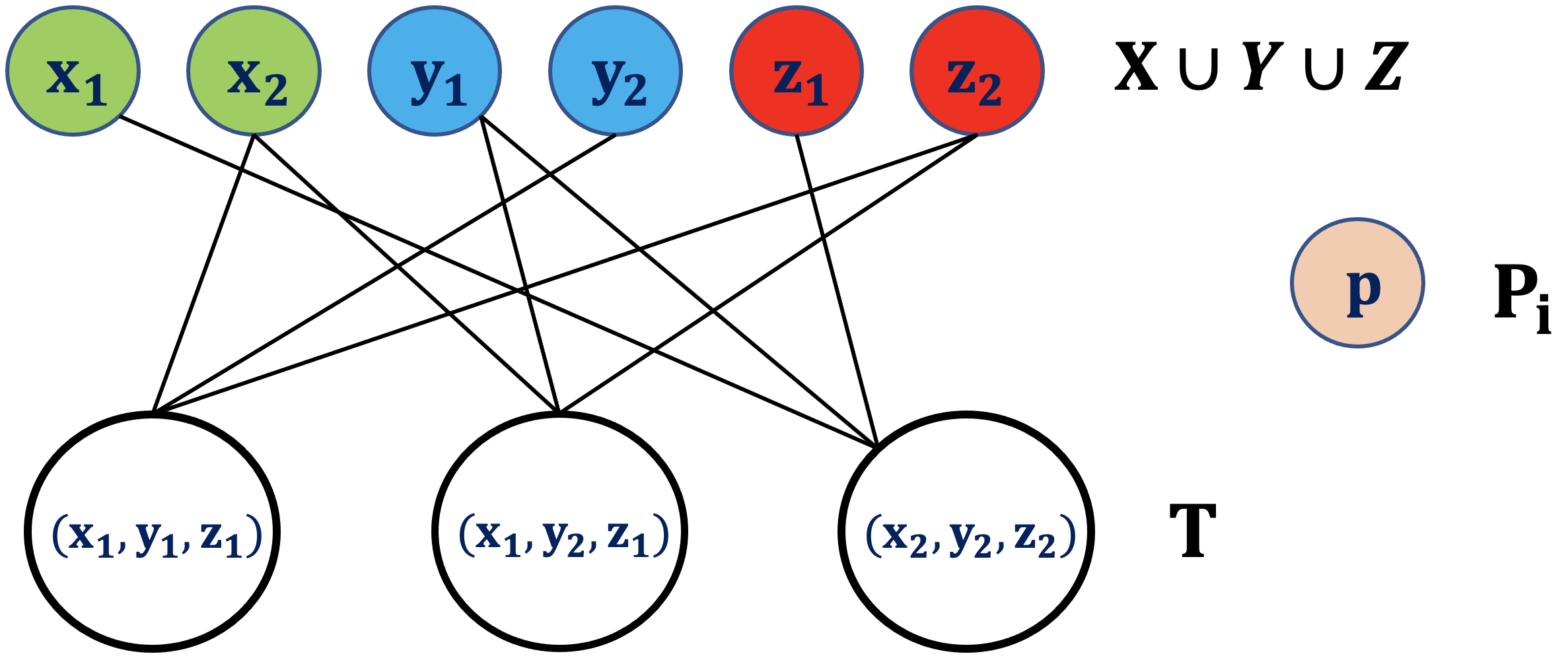}%
	\hspace{1mm}%
	\caption{An illustration of the BPB instance $\ci_i = (I_i,s,E)$, where $i = \OPT(\cj) = 2$. The optimal solution for $\ci_i$ contains the bins $\{x_1,y_1,z_1,(x_1,y_1,z_1)\}, \{x_2,y_2,z_2,(x_2,y_2,z_2)\}$, and $\{p,(x_1,y_2,z_1)\}$; this corresponds to an optimal solution $(x_1,y_1,z_1), (x_2,y_2,z_2)$ for the original B3DM instance. Note that in this example $Q_i = \emptyset$.}
	\label{fig:proof}
\end{figure}

Given a packing $(A_1, \ldots,A_q)$ for the BPB instance $\ci_i$, we consider all {\em useful bins} $A_b$ in the packing, i.e., $A_b = \{x,y,z,t\}$, where $x \in X, y\in Y,z \in Z$ and $t = (x,y,z)$. The triple $t$ from bin $A_b$ is taken to our solution for the original $n$-restricted B3DM instance $\cj$. Note that taking all triples as described above forms a feasible solution for $\cj$, since each element is packed only once. Thus, our goal becomes to find a packing for the reduced BPB instance with a maximum number of useful bins.
Indeed, since $s(A_b) = 1$ for any useful bin $A_b$, finding a packing with many useful bins coincides with an efficient approximation for BPB. 

For the optimal guess $i^* = \OPT(\cj)$, it is not hard to see that the optimum for the BPB instance $\cI_{i^*}$ satisfies $s(I_{i^*}) = \OPT(\ci_{i^*})$; that is, all bins in the optimum are {\em fully} packed. For a sufficiently large $n$, and assuming there is an APTAS for BPB, we can find a packing of $\cI_{i^*}$ with a large number of bins that are fully packed. A majority of these bins are useful, giving an efficient approximation for the original B3DM instance. A similar reduction to BPS is obtained by 
adding to the bipartite conflict graph of the BPB instance an edge between any pair of vertices in $T$; thus, we have a {\em split} conflict graph. 
We summarize the above discussion
 in the next result (the proof is given in Appendix~\ref{sec:proofsBPHardness}). 
\begin{theorem}
	\label{thm:BPH}
	There is no \textnormal{APTAS} for \textnormal{BPB} and \textnormal{BPS}, unless \textnormal{P=NP}. 
\end{theorem}

\mysection{Discussion}
\label{sec:discussion}

In this work we presented the first theoretical evidence that BPC
on polynomially colorable graphs is harder than classic bin packing, 
even in the special cases of bipartite and split graphs. 
Furthermore, we introduced a new generic framework for tackling 
BPC instances, based on a reduction to a maximization problem. Using this framework, we improve the state-of-the-art approximations for BPC on several well studied graph classes.  

We note that better bounds for the maximization problems solved within our framework will imply improved approximation guarantees for BPC on perfect, bipartite, and split graphs. It would be interesting to apply our techniques to improve the known results 
for other graph classes, such as chordal graphs or partial $k$-trees. 

\bibliographystyle{splncs04}
\newpage
\bibliography{bibfile}
\appendix

\section{Related Work}

\label{sec:related}

The BPC problem
was introduced by Jansen and  {\"{O}}hring~\cite{JO97}. They presented a general algorithm that initially 
finds a coloring of the conflict graph, and then packs each color class separately using the First-Fit Decreasing algorithm. This approach yields a $2.7$-approximation for BPC on perfect graph. The paper
\cite{JO97} includes also a $2.5$-approximation for subclasses of perfect graphs on which the corresponding {\em precoloring extension problem} can be solved in polynomial time (e.g., interval and chordal graphs). 
The authors present also a $(2+\eps)$-approximation algorithm for BPC on cographs and partial $k$-trees.

Epstein and Levin~\cite{epstein2008bin} present better algorithms for BPC on perfect graphs ($2.5$-approximation), graphs on which the precoloring extension problem can be solved in polynomial time ($\frac{7}{3}$-approximation), and bipartite graphs ($\frac{7}{4}$-approximation). Their techniques include matching between {\em large} items and 
a sophisticated use of new item {\em weights}. Recently, Huang et al. \cite{huang2023improved} 
provided fresh insights to previous algorithms, leading to  $\frac{5}{3}$-approximation for BPC on bipartite graphs and a $2$-approximation on  split graphs.

Jansen~\cite{jansen1999approximation} presented an AFPTAS for BPC on {\em d-inductive} conflict graphs,
where $d \geq 1$ is some constant. This graph family includes  trees, grid graphs, planar graphs, and graphs with constant treewidth. For a survey of {\em exact} algorithms for BPC see, e.g., \cite{huang2023improved}. 

\section{Omitted Proofs from Section~\ref{sec:preliminaries}}
\label{sec:proofsPrel}

	For completeness, we describe the well-known algorithm \textsf{FFD} with our notation, and give some properties of the algorithm that will be useful for our analysis. Given a BP instance $\ci = (I,s)$, in the \textsf{FFD} Algorithm the items are considered in a non-increasing order $v_1, \ldots, v_n$ by their sizes. Starting with an empty packing, in each iteration \textsf{FFD} attempts to assign the largest unpacked item $\ell_i$ to an open bin; if none of these bins can accommodate the item, the algorithm opens a new bin to which  the item is assigned. The pseudocode is given in Algorithm~\ref{alg:FFD}.

	 \begin{algorithm}[h]
		\caption{$\textsf{FFD}(\ci = (I,s))$}
		\label{alg:FFD}
			\begin{algorithmic}[1]
			\State{Let $v_1, \ldots,v_n$ be the items in $I$ in non-increasing order by sizes.\label{step:FFDsort}}
		
			\State{Initialize an empty packing $\cB \leftarrow ()$.\label{step:FFDinit}} 
		
		\For{$i \in [n]$}

			\If{$\exists B \in \cB \textnormal{ s.t. } s(B)+s(v_i) \leq 1$\label{step:FFDif}}
				
					\State{$B \leftarrow B \cup \{v_i\}$.\label{step:FFDnotnew}}
				
			\EndIf
				
				\State{\textbf{else} $\cB \leftarrow \cB \oplus \{v_i\}$\label{step:FFDnew}}

				\State{Return $\cB$.\label{step:FFDRet}}
			
		\EndFor
			\end{algorithmic}
	\end{algorithm}

	We give below several claims relating to the approximation guarantee of Algorithm \textsf{FFD}. Given a BP instance $\cI = (I,s)$ and $0<\eps<0.1$, let $T_{\ci}(\eps) = \{v \in I~|~s(v) \leq \eps\}$ be the set of {\em tiny} items and let $B_{\ci}(\eps) = I \setminus T_{\ci}$ be the set of {\em big} items of $\ci$ and $\eps$. When understood from the context, we simply use $T_{\ci} = T_{\ci}(\eps)$ and $B_{\ci} = B_{\ci}(\eps)$.
	
	 \begin{lemma}
		\label{claim:4}
		Let $\ci = (I,s)$ be a \textnormal{BP} instance and let $0<\eps<0.1$ such that $s(T_{\ci})>1$ and $s(B_{\ci}) < 2$. Then, $(1-\eps) \cdot (\#\textnormal{\textsf{FFD}}(\ci)-1) \leq s(I)$.  
	\end{lemma}
	\begin{proof}
		Let $v_1, \ldots, v_n$ be the order of $I$ obtained by Step~\ref{step:FFDsort}. In addition, let $\textnormal{\textsf{FFD}}(\ci) = (A_1, \ldots,A_t)$ such that $A_1, \ldots, A_t$ is the order in which $\textnormal{\textsf{FFD}}(\ci)$ creates new bins in Step~\ref{step:FFDnew}. By the item ordering, we may assume that all bins in the obtained packing that contain a large item appear as a prefix of the returned packing. Thus, let $A_1, \ldots, A_r$ be all bins in $\textnormal{\textsf{FFD}}(\ci)$ that contain a big item.

		\begin{myclaim}
		\label{claim:r<4}
		$r \leq 3$.
	\end{myclaim}
	\begin{proof}
		Assume towards a contradiction that $r>3$ and consider the last item $v \in B_{\ci}$ packed in $\textnormal{\textsf{FFD}}(\ci)$; let $i \in [n]$ such that $v = v_i$. We consider two cases. \begin{itemize}
			\item  $s(v) > \frac{1}{2}$. \begin{equation}
				\label{eq:claim4Cont1}
				s(B_{\ci}) \geq \sum_{j \in [i]} s(v_j) \geq s((A_1+A_2+A_3+A_4) \cap B_{\ci}) > 4 \cdot \frac{1}{2} = 2. 
			\end{equation}

		The second inequality holds since each of the bins $A_1, \ldots, A_r$ contains a big item and due to the ordering of the items. Thus, we get a contradiction to $s(B_{\ci}) < 2$. 
		
			\item $s(v) \leq \frac{1}{2}$. By Step~\ref{step:FFDif} it holds that $s(A_k \cap B_{\ci})+s(v) > 1,  \forall k \in [r-1]$; thus,  \begin{equation}
				\label{eq:claim4Cont2}
				s(B_{\ci}) \geq (r-1) \cdot (1-s(v))+s(v) \geq 3-2 \cdot s(v) \geq 2.
			\end{equation}  The second inequality holds since $r >3$. The last inequality holds since $s(v) \leq \frac{1}{2}$. By \eqref{eq:claim4Cont2}, we reach a contradiction that $s(B_{\ci}) < 2$. 
		\end{itemize}
	\qed \end{proof}

T complete the proof of the lemma we need the following claims. 
	\begin{myclaim}
	\label{claim:r2}
	If $t > r$, then for all $i \in[t-1]$ it holds that $s(A_i) \geq 1-\eps$. 
\end{myclaim}
\begin{proof}
	Consider the first item $v$ packed in bin $A_t$ in Step~\ref{step:FFDnew}. Because $t > r$ it holds that $v \in T_{\ci}$. Therefore, by Step~\ref{step:FFDif}  for all $i \in [t-1]$ we have $s(A_i) +s(v) > 1$ and the claim follows. 
\qed \end{proof}

	\begin{myclaim}
	\label{claim:r3}
	If $t = r = 3$, then $s(A_1 \cap B_{\ci})+s(A_2 \cap B_{\ci}) \geq 1$. 
\end{myclaim}
\begin{proof}
	Consider the first item $v$ packed in bin $A_2$ in Step~\ref{step:FFDnew} and let $A'_1$ be the set of items packed in $A_1$ in the iteration in which the algorithm packs $v$. Therefore, 
	
	$$s(A_1 \cap B_{\ci})+s(A_2 \cap B_{\ci}) \geq s(A'_1) +s(v) > 1$$ 
	
	The last inequality follows by the algorithm.

\qed \end{proof}

		  By Claim~\ref{claim:r<4} it holds that $r\leq 3$. Thus, we we consider three cases. \begin{enumerate}
			\item $t = r = 2$.  Then, 
			\begin{equation*}
				\label{eq:t=r=2}
				s(I) \geq s(T_{\ci}) > 1 \geq (1-\eps) \cdot (t-1) = (1-\eps) \cdot (\#\textnormal{\textsf{FFD}}(\ci)-1).
			\end{equation*}
		\item $t = r = 3$. Then, 
		
		\begin{equation*}
			\label{eq:t=r=3}
			\begin{aligned}
				s(I) \geq{} & s(A_1 \cap B_{\ci})+s(A_2 \cap B_{\ci})+s(T_{\ci}) \\
				>{} & 1+1 \\
				\geq{} & (1-\eps) \cdot (t-1) \\
				={} & (1-\eps) \cdot (\#\textnormal{\textsf{FFD}}(\ci)-1).		
			\end{aligned}
		\end{equation*} The second inequality holds by Claim~\ref{claim:r3}.

		\item $r < t$. Then,  \begin{equation*}
			\label{eq:t>r}
				s(I) \geq (1-\eps) \cdot (t-1) = (1-\eps) \cdot (\#\textnormal{\textsf{FFD}}(\ci)-1).
		\end{equation*} The inequality holds by Claim~\ref{claim:r2}. 
		
	\end{enumerate}
	
	This completes the proof of the lemma. $\blacksquare$
\end{proof}
	
	The next observation follows since if $s(I) \leq \frac{3}{2}$ for some BP instance $\ci = (I,s)$ then, clearly, the first two bins $A_1, A_2$ in $\textsf{FFD}(\ci)$ satisfy $s(A_1)+s(A_2) > 1$. Moreover, $s(A_1) \geq \frac{1}{2}$ by the ordering of the items and it follows that $A_2$ can pack all remaining items. 
	\begin{observation}
		\label{lem:FFD:one}
			Let $\ci = (I,s)$ be a \textnormal{BP} instance such that $s(I) \leq \frac{3}{2}$. Then, $\#\textnormal{\textsf{FFD}}(\ci) \leq 2$. 
	\end{observation}

	\noindent {\bf Proof of Lemma~\ref{lem:FF}:}
 With a slight abuse of notation, given a packing $\cB = (B_1, \ldots,B_t)$, we say that $B \in \cB$ if there is $i \in [t]$ such that $B = B_i$; otherwise, we say that $B \notin \cB$.
	Let $I = \{v_1, \ldots, v_n\}$ be the order used by \textnormal{\textsf{FFD}}. Also, let $\cB = (B_1, \ldots,B_t)$ be the packing returned by $\textnormal{\textsf{FFD}}(\ci)$. 
		\begin{myclaim}
		\label{claim:ijt}
		There is at most one $i \in [t]$ such that $s(A_i)< 1-\max_{\ell \in I} s(\ell)$.
	\end{myclaim}
	\begin{proof}
		Assume towards a contradiction that there are $i,j \in [t], i < j$ such that $s(A_i)< 1-\max_{\ell \in I} s(\ell)$ and $s(A_j)< 1-\max_{\ell \in I} s(\ell)$. Consider the first item $v$ that is added to $B_j$; by 
		the \textsf{FFD} algorithm, the item should be added to $B_i$ or some other available bin and should not be added in a new bin $B_j$. Contradiction.
	\qed \end{proof}

For the following, we define {\em weights} for the elements as introduced in \cite{epstein2008bin}. For every $v \in L_{\ci}$ define $w(v) = 1$; in addition, for every $v \in M_{\ci}$ define $w(v) = s(v)+\frac{1}{6}$; finally, for all $v \in S_{\ci}$ define 	$w(v) = s(v)+\frac{1}{12}$. The next result follows by a result of \cite{epstein2008bin}.\footnote{The result in \cite{epstein2008bin} is stronger since their weights refine our definition of weights of small items.} 	\begin{myclaim}
	\label{claim:epsh:2}
	$\#\cB \leq w(I)+1$. 
\end{myclaim}

		\begin{myclaim}
		\label{claim:LMS}
		$\#\cB \leq |L_{\ci}|+\frac{3}{2} \cdot s(M_{\ci})+\frac{4}{3} \cdot s(S_{\ci})+1$.
	\end{myclaim}
	\begin{proof}

		Let $ k\in [n]$ such that $v_k \in I$ is the first item packed in bin $B_t$. By Step~\ref{step:FFDif}, for all $i \in [t-1]$ it holds that 
		
		\begin{equation}
			\label{eq:vvv}
			s(A_i)+s(v_k)>1.
		\end{equation} 
	
	We consider two cases. \begin{itemize}
		\item $s(v_k) \leq \frac{1}{4}$ \begin{equation}
			\label{eq:sB2}
			\begin{aligned}
				\#\cB \leq{} &  |L_{\ci}|+\frac{4}{3} \cdot s(I \setminus L_{\ci})+1 \leq |L_{\ci}|+ \frac{3}{2} \cdot s(M_{\ci})+\frac{4}{3} \cdot s(S_{\ci})+1
			\end{aligned}
		\end{equation} The first inequality holds because there can be at most $|L_{\ci}|$ bins with a large item; moreover, there can be at most $\frac{4}{3} \cdot s(I \setminus L_{\ci})+1$ bins without a large item by \eqref{eq:vvv} (each bin is at least $\frac{3}{4}$-full).

		\item $s(v_k) > \frac{1}{4}$. Let $I_k =  \{v_1, \ldots, v_k\}$ and let $\ci_k = (I_k,s)$. By Algorithm~\ref{alg:FFD} it holds that
		
		 \begin{equation}
			\label{eq:vk}
			\#\cB = \#\textnormal{\textsf{FFD}}(\cI_k).
		\end{equation} 
	
	The reason is that all items $v_{k+1}, \ldots, v_n$ are packed in the existing bins $B_1, \ldots, B_t$ by the definition of $v_k$. 
	
	\begin{equation}
			\label{eq:nnnn}
			\begin{aligned}
				\#\cB ={} & \#\textnormal{\textsf{FFD}}(\cI_k) \\
				 \leq{} & w(I_k)+1 \\
				={} & w(L_{\ci_k})+w(M_{\ci_k})+w(S_{\ci_k})+1\\
				={} & |L_{\ci_k} |+s(M_{\ci_k})+\frac{|M_{\ci_k}|}{6}+s(S_{\ci_k})+\frac{|S_{\ci_k}|}{12}+1\\
				\leq{} & |L_{\ci_k}|+s(M_{\ci_k})+\frac{3 \cdot s(M_{\ci_k})}{6}+s(S_{\ci_k})+\frac{4 \cdot s(S_{\ci_k})}{12}+1\\
					\leq{} & |L_{\ci}|+s(M_{\ci})+\frac{3 \cdot s(M_{\ci})}{6}+s(S_{\ci})+\frac{4 \cdot s(S_{\ci})}{12}+1\\
				={} & |L_{\ci}|+\frac{3}{2} \cdot s(M_{\ci})+\frac{4}{3} \cdot s(S_{\ci})+1.
			\end{aligned}
		\end{equation} The first equality follows by \eqref{eq:vk}. The first inequality holds by Claim~\ref{claim:epsh:2}. The second equality holds by the definition of weights. 
	\end{itemize}

 \end{proof} 	\qed

The statement of the Lemma follows from Claim~\ref{claim:ijt} and Claim~\ref{claim:LMS}. $\blacksquare$

~\\
	{\bf Proof of Lemma~\ref{lem:rothvos}:}
	Let \textsf{APTAS} and \textsf{AdditiveBP} be the asymptotic PTAS for BP by \cite{fernandez1981bin} and the additive approximation scheme for BP by \cite{hoberg2017logarithmic}, respectively. Define the following algorithm which returns the better packing resulting from the two algorithms, where we use $\eps = 0.001$ for  \textsf{APTAS}. We give the pseudocode in Algorithm~\ref{alg:asymBP}.

	 \begin{algorithm}[h]
		\caption{$\textsf{AsymptoticBP}(\ci)$}
		\label{alg:asymBP}
			\begin{algorithmic}[1]
		\State{Let $\eps = 0.001$.\label{step:aseps}}
		
			\State{Compute $\cA_1 \leftarrow \textsf{AdditiveBP}(\ci)$, $\cA_2 \leftarrow \textsf{APTAS}(\ci, \eps)$.\label{step:asc}}
		
			\State{Return $\argmin_{\cA \in \{\cA_1, \cA_2\}} \#\cA$.\label{step:asr}}
			\end{algorithmic}
	\end{algorithm} 

Then, by \cite{hoberg2017logarithmic}, Step~\ref{step:asc}, and  Step~\ref{step:asr} of Algorithm~\ref{alg:asymBP}, it holds that algorithm \textnormal{\textsf{AsymptoticBP}} returns a packing $\cB = (B_1, \ldots,B_t)$ of $\ci$ such that $t = \OPT(\ci)+o(\OPT(\ci))$. In addition, if $\OPT(\ci) \geq 100$, then by \cite{fernandez1981bin}, Step~\ref{step:asc}, and  Step~\ref{step:asr} of Algorithm~\ref{alg:asymBP}, it holds that  algorithm \textnormal{\textsf{AsymptoticBP}} returns a packing $\cB = (B_1, \ldots,B_t)$ of $\ci$ such that $$t \leq (1+\eps) \cdot \OPT(\ci)+1 \leq 1.001 \cdot \OPT(\ci)+1 <  1.02 \cdot \OPT(\ci).$$ The first inequality holds by \cite{fernandez1981bin}. The second inequality holds since $\eps = 0.001$. The last inequality holds since $\OPT(\ci) \geq 100$. $\blacksquare$
\section{Omitted Proofs from Section~\ref{sec:BP}}
\label{sec:proofsBP}

We use in the proof of Lemma~\ref{lem:Generic} the next two lemmas.

\begin{lemma}
	\label{lem:Generic1}
	Given a \textnormal{BPC} instance $\ci = (I,s,E)$, Algorithm \textnormal{\textsf{\Generic}} returns  a packing $\cB$ of $\ci$ such that $\#\cB \leq \chi(G_{\ci})+|L_{\ci}|+\frac{3}{2} \cdot s(M_{\ci})+\frac{4}{3} \cdot s(S_{\ci})$. 
\end{lemma}
\begin{proof}

		Let $\cC = (C_1, \ldots,C_k)$ be the  minimal coloring of $G_{\cI}$ found in Step~\ref{step:GenericCol} of Algorithm~\ref{alg:generic}. Now, \begin{equation*}
			\begin{aligned}
				\#\cB  \leq{} &  \sum_{i \in [k]}  \#\textsf{\textsf{FFD}}((C_i,s)) \\
				\leq{} &  \sum_{i \in [k]} \left(|L_{\ci} \cap C_i|+\frac{3}{2} \cdot s(M_{\ci} \cap C_i)+\frac{4}{3} \cdot s(S_{\ci} \cap C_i)+1\right) \\
				={} &  \chi(G_{\ci})+|L_{\ci}|+\frac{3}{2} \cdot s(M_{\ci})+\frac{4}{3} \cdot s(S_{\ci}).
			\end{aligned}
		\end{equation*}

		The first inequality follows from the algorithm. 
		The second inequality holds by Lemma~\ref{lem:FF}. The last equality follows because the sets in the coloring $\cC$ are disjoint and complementary; also, $k = \chi(G_{\ci})$ because $\cC$ is a minimal coloring of $G_{\ci}$. $\blacksquare$
	\end{proof}

	\begin{lemma}
	\label{lem:Generic2}
	Given a \textnormal{BPB} instance $\ci = (I,s,E)$, Algorithm \textnormal{\textsf{\Generic}} returns a packing $\cB$ of $\ci$ such that $\#\cB \leq \frac{3}{2} \cdot |L_{\ci}|+\frac{4}{3} \cdot \left(\OPT(\ci)-|L_{\ci}|\right)+o(\OPT(\ci))$. 
\end{lemma}
\begin{proof}
	
	Recall that $X_V, Y_V$ denote the bipartition of the given conflict graph. Let $\OPT = (O_1, \ldots,O_t)$ be an optimal packing of $\ci$. We define below a partition of $O_1, \ldots, O_t$ into {\em types} of bins, based on the number of large and medium items from $X_V$ and $Y_V$; the partition also considers the distinct cases where the bins contain total size larger from $X_V$ or from $Y_V$. Let $\OPT_X = \{O_i~|~i \in [t], s(O_i \cap X_V) \geq s(O_i \cap Y_V)\}$ be all bins in which the total size of items packed from $Y_V$ is at most the total size of items packed from $X_V$. 
	Similarly, let  $\OPT_Y = \{O_i~|~i \in [t], s(O_i \cap X_V) < s(O_i \cap Y_V)\}$ be all remaining bins in $\OPT$. We now refine this partition. Let $W_V \in \{X_V,Y_V\}$ and let $\bar{W}_V$ be $Y_V$ if $W_V = X_V$ and $X_V$ if $W_V = Y_V$.  As our definitions are symmetric for $X_V$ and $Y_V$, we define them w.r.t. $W_V$. Let
	
	 \begin{equation}
		\label{eqbinTypesW}
			T_{L}(W) = \{O_i \in \OPT_W~|~O_i \cap L_{\ci} \cap W_V \neq \emptyset\}
	\end{equation} 

be all bins in $\OPT$ that contains a large item from $W_V$. For the following, we use the abbreviations $\alpha = \OPT_W\setminus T_{L}(W)$, $N_{i,M} =|O_i \cap M_{\ci} \cap W_V|$, and $\bar{N}_{i,M} = |O_i \cap M_{\ci} \cap \bar{W}_V|$, where $N_{i,M}$, and $\bar{N}_{i,M}$ are the number of medium elements in bin $O_i$ from $W_V, \bar{W}_V$, respectively. Now define the following partition of the bins of $\OPT_W$ by all possible values of $N_{i,M}$, and $\bar{N}_{i,M}$: 

\begin{equation}
		\label{eq:binTypes}
		\begin{aligned}
			T_{1,0,\leq}(W) ={} & \{O_i \in \alpha~|~N_{i,M}= 1, \bar{N}_{i,M}= 0, s(O_i \cap \bar{W}_V) \leq \frac{1}{3}\} \\
			T_{1,0,>}(W) ={} & \{O_i \in \alpha~|~N_{i,M}= 1,\bar{N}_{i,M} = 0, s(O_i \cap \bar{W}_V) > \frac{1}{3}\} \\
			T_{1,1}(W) ={} & \{O_i \in \alpha~|~N_{i,M} = 1, \bar{N}_{i,M}= 1\} \\
			T_{0,0}(W) ={} & \{O_i \in \alpha~|~N_{i,M}= 0, \bar{N}_{i,M}= 0\} \\
			T_{0,1}(W) ={} & \{O_i \in \alpha~|~N_{i,M} = 0,\bar{N}_{i,M}= 1\} \\
			T_{2,0}(W) ={} & \{O_i \in \alpha~|~N_{i,M}= 2,\bar{N}_{i,M}= 0\} \\
		\end{aligned}
	\end{equation}

	For example, $T_{1,0,\leq}(W) $ is the set of all bins in $\OPT_W \setminus T_L(W) = \alpha$ that contain one medium item from $W_V$, no medium items from $\bar{W}_V$ and the total size from $\bar{W}$ is at most $\frac{1}{3}$. We now define a packing of $\ci$ such that each bin type is packed separately, where items from $X_V$ and $Y_V$ are also packed separately. We now define a packing for the items in each bin type. Fix some $W \in \{X,Y\}$ and let $T_L(W) = (A_1, \ldots,A_{r})$; by adding at most $1$ empty bin to $T_L(W)$, we may assume from now on that $r  \in \mathbb{N}_{\textnormal{even}}$. Define
	
	 \begin{equation}
		\label{eq:p1}
		\cB_L(W) = \left(A_i \cap W_V~|~i \in [r]\right) \oplus \left((A_i \cup A_{i+1}) \cap \bar{W}_V~|~i \in \{1,3, \ldots, r-1\}\right)
	\end{equation}

\begin{myclaim}
		\label{claim:p1}
		$\cB_L(W)$ is a packing of $\textnormal{\textsf{items}}(T_L(W))$ such that $\#\cB_L(W) \leq \frac{3}{2} \cdot |T_L(W)|+1$.
	\end{myclaim}

	\begin{proof}
		For all $i \in [r]$ and $j \in \{1,3,\ldots, r-1\}$ it holds that $A_i \cap W_V \subseteq W_V$ and $(A_j \cup A_{j+1}) \cap \bar{W}_V \subseteq \bar{W}_V$; thus, $A_i \cap W_V$ and $(A_j \cup A_{j+1}) \cap \bar{W}_V$ are independent sets in $G$. Moreover, it holds that $s(A_i \cap W_V) \leq s(A_i) \leq 1$, where the second inequality holds since $\OPT$ is a packing of $\ci$. Finally, it holds that

		\begin{equation*}
			\label{eq:sB2}
			\begin{aligned}
			s((A_j \cup A_{j+1}) \cap \bar{W}_V) \leq{} & s(A_{j})+s(A_{j+1})-s(A_j \cap L_{\ci} \cap W_V)-s(A_{j+1} \cap L_{\ci} \cap W_V) \\
			\leq{} & 1+1-\frac{1}{2}-\frac{1}{2} \\
			={} & 1.
			\end{aligned}
		\end{equation*}

		 The second inequality holds since $\OPT$ is a packing of $\ci$ and since $A_j \cap L_{\ci} \cap W_V \neq \emptyset$ and $A_{j+1} \cap L_{\ci} \cap W_V \neq \emptyset$ by \eqref{eq:binTypes}. It follows that $\cB_L(W)$ is a packing of $\textsf{items}(T_L)$. By \eqref{eq:p1} it holds that $\#\cB_L(W) \leq |T_L(W)|+ \frac{1}{2} \cdot |T_L(W)|+1 =  \frac{3}{2} \cdot |T_L(W)|+1$.
	\qed \end{proof}

	Let $Z = \{\{1,0,\leq\},\{1,0,>\},\{1,1\},\{0,0\},\{0,1\},\{2,0\}\}$. For the simplicity of the notations, in the following we use $A^z_1, \ldots, A^z_{r_z}$ to denote the bins in $T_z(W)$, for $z \in Z$; when understood from the context, we simply use $A_i$ for $A^z_i$ for any $i \in [r_z]$. For simplicity, note that by adding at most $5$ empty bins to $T_z(W)$, we assume from now on that $r_z = 6 \cdot n_z, n_z \in \mathbb{N}$ where 
	
	\begin{equation}
		\label{eq:zzz}
		r_z \leq |T_z(W)|+5.
	\end{equation} 

Now, define

\begin{equation}
	\label{eq:p2}
	\begin{aligned}
		\cB_{\{1,0,\leq\}}(W) ={} & \left(A_i \cap W_V~|~i \in [r_z]\right) \oplus \\
	{} & \left((A_i \cup A_{i+1} \cup A_{i+2}) \cap \bar{W}_V~|~i \in \{1,4, \ldots, r_z-2\}\right)
	\end{aligned}
\end{equation}

\begin{myclaim}
		\label{claim:p2}
		Let $z = \{1,0,\leq\}$. Then, $	\cB_{z}(W)$ is a packing of $\textnormal{\textsf{items}}(T_z(W))$ such that
		
		 $$\#	\cB_{z}(W) \leq \frac{4}{3} \cdot |T_{z}(W)|+7.$$

	\end{myclaim}
	\begin{proof}
		For all $i \in [r_z]$ and $j \in  \{1,4, \ldots, r_z-2\}$ it holds that $A_i \cap W_V \subseteq W_V$ and $(A_j \cup A_{j+1}  \cup A_{j+2}) \cap \bar{W}_V \subseteq \bar{W}_V$; thus, $A_i \cap W_V$ and $(A_j \cup A_{j+1}  \cup A_{j+2}) \cap \bar{W}_V$ are independent sets in $G$. Moreover, it holds that $s(A_i \cap W_V) \leq s(A_i) \leq 1$, where the second inequality holds since $\OPT$ is a packing of $\ci$. Finally, it holds that $s((A_j \cup A_{j+1}  \cup A_{j+2}) \cap \bar{W}_V) \leq \frac{1}{3} \cdot 3 = 1$. The inequality holds by \eqref{eq:binTypes}. It follows that $\cB_L(W)$ is a packing of $\textsf{items}(T_L)$. By \eqref{eq:p2} it holds that $\#\cB_z(W) \leq r_z+ \frac{1}{3} \cdot r_z =  \frac{4}{3} \cdot r_z \leq \frac{4}{3} \cdot |T_{z}(W)|+7$.
	\qed \end{proof}

We define
	\begin{equation}
		\label{eq:p3}
		\begin{aligned}
			\cB_{\{1,0,>\}}(W) ={} & \left((A_i \cup A_{i+1}) \cap M_{\ci}~|~i \in \{1,3, \ldots, r_z-1\}\right) \\ \oplus{} &\left((A_i \cup A_{i+1}) \cap \bar{W}_V~|~i \in \{1,3, \ldots, r_z-1\}\right) \\ \oplus{} & \left((A_i \cup A_{i+1} \cup A_{i+2}) \cap (W_V \setminus M_{\ci}) ~|~i \in \{1,4, \ldots, r_z-2\}\right) 
		\end{aligned}
	\end{equation}\begin{myclaim}
		\label{claim:p3}
		Let $z = \{1,0,>\}$. Then, $\cB_{z}(W)$ is a packing of $\textnormal{\textsf{items}}(T_z(W))$ such that $$\#	\cB_{z}(W) \leq \frac{4}{3} \cdot |T_{z}(W)|+7.$$
	\end{myclaim}
	\begin{proof}
		
		For all $i \in \{1,3, \ldots, r_z-1\}$, $j \in  \{1,3, \ldots, r_z-1\}$, and $k \in \{1,4, \ldots, r_z-2\}$ it holds that $(A_i \cup A_{i+1}) \cap M_{\ci} \subseteq W_V$ by \eqref{eq:binTypes}, $(A_j \cup A_{j+1}) \cap \bar{W}_V \subseteq \bar{W}_V$, and $(A_k \cup A_{k+1}  \cup A_{k+2}) \cap (W_V \setminus M_{\ci}) \subseteq W_V$. Thus, $(A_i \cup A_{i+1}) \cap M_{\ci}$, $(A_j \cup A_{j+1}) \cap \bar{W}_V$, and $(A_k \cup A_{k+1}  \cup A_{k+2}) \cap (W_V \setminus M_{\ci})$ are independent sets in $G$.
		Moreover, it holds that $s((A_i \cup A_{i+1}) \cap M_{\ci}) \leq 2 \cdot\frac{1}{2} \leq 1$ by the definition of medium items; also, $s((A_{j} \cup A_{j+1}) \cap \bar{W}_V) \leq 2 \cdot \frac{1}{2}$ by \eqref{eq:binTypes}. Finally, $s((A_k \cup A_{k+1}  \cup A_{k+2}) \cap (W_V \setminus M_{\ci})) \leq 3 \cdot \frac{1}{3}$, where the inequality follows because for all $t \in \{k,k+1,k+2\}$ it holds that $s(A_t) - s(A_t \cap M_{\ci}) - s(A_t \cap \bar{W}_V) \leq 1- \frac{1}{3}-\frac{1}{3} = \frac{1}{3}$ by \eqref{eq:binTypes}. By \eqref{eq:p3} it holds that $\#\cB_z(W) \leq \frac{r_z}{2}+\frac{r_z}{2}+\frac{r_z}{3} \leq  \frac{4}{3} \cdot r_z \leq \frac{4}{3} \cdot |T_{z}(W)|+7$.
		
	\qed \end{proof}
	Define

	\begin{equation}
		\label{eq:p4}
		\begin{aligned}
			\cB_{\{1,1\}}(W) ={} & \left((A_i \cup A_{i+1}) \cap (M_{\ci} \cap W_V)~|~i \in \{1,3, \ldots, r_z-1\}\right) \\ \oplus{} &\left((A_i \cup A_{i+1}) \cap \bar{W}_V~|~i \in \{1,3, \ldots, r_z-1\}\right) \\ \oplus{} & \left((A_i \cup A_{i+1} \cup A_{i+2}) \cap (W_V \setminus M_{\ci}) ~|~i \in \{1,4, \ldots, r_z-2\}\right) 
		\end{aligned}
	\end{equation}\begin{myclaim}
		\label{claim:p4}
		Let $z = \{1,1\}$. Then, $\cB_{z}(W)$ is a packing of $\textnormal{\textsf{items}}(T_z(W))$ such that $$\#	\cB_{z}(W) \leq \frac{4}{3} \cdot |T_{z}(W)|+7.$$
	\end{myclaim}
	\begin{proof}
		For all $i \in \{1,3, \ldots, r_z-1\}$, $j \in  \{1,3, \ldots, r_z-1\}$, and $k \in \{1,4, \ldots, r_z-2\}$ it holds that $(A_i \cup A_{i+1}) \cap (M_{\ci} \cap W_V) \subseteq W_V$, $(A_j \cup A_{j+1}) \cap \bar{W}_V \subseteq \bar{W}_V$, and $(A_k \cup A_{k+1}  \cup A_{k+2}) \cap (W_V \setminus M_{\ci}) \subseteq W_V$. Thus, $(A_i \cup A_{i+1}) \cap (M_{\ci} \cap W_V)$, $(A_j \cup A_{j+1}) \cap \bar{W}_V$, and $(A_k \cup A_{k+1}  \cup A_{k+2}) \cap (W_V \setminus M_{\ci})$ are independent sets in $G$.
		Moreover, it holds that $s((A_i \cup A_{i+1}) \cap (M_{\ci} \cap W_V)) \leq 2 \cdot\frac{1}{2} \leq 1$ by the definition of medium items; also, $s((A_{j} \cup A_{j+1}) \cap \bar{W}_V) \leq 2 \cdot \frac{1}{2}$ by \eqref{eq:binTypes}. Finally, $s((A_k \cup A_{k+1}  \cup A_{k+2}) \cap (W_V \setminus M_{\ci})) \leq 3 \cdot \frac{1}{3}$, where the inequality follows because for all $t \in \{k,k+1,k+2\}$ it holds that $s(A_t) - s(A_t \cap (M_{\ci} \cap W_V)) - s(A_t \cap \bar{W}_V) \leq 1- \frac{1}{3}-\frac{1}{3} = \frac{1}{3}$ by \eqref{eq:binTypes}. By \eqref{eq:p4} it holds that $\#\cB_z(W) \leq \frac{r_z}{2}+\frac{r_z}{2}+\frac{r_z}{3} \leq  \frac{4}{3} \cdot r_z \leq \frac{4}{3} \cdot |T_{z}(W)|+7$.
		
	\qed \end{proof}

	Define
	\begin{equation}
		\label{eq:p5}
		\begin{aligned}
			\cB_{\{0,0\}}(W) ={} & \textsf{\Generic}\left(\ci \cap \bigcup_{i \in [r_z]} A_i \right). 
		\end{aligned}
	\end{equation}\begin{myclaim}
		\label{claim:p5}
		Let $z = \{0,0\}$. Then, $\cB_{z}(W)$ is a packing of $\textnormal{\textsf{items}}(T_z(W))$ such that $$\#	\cB_{z}(W) \leq \frac{4}{3} \cdot |T_{z}(W)|+9.$$
	\end{myclaim}
	\begin{proof}
		Let $\cj = \ci \cap \bigcup_{i \in [r_z]} A_i$. By Lemma~\ref{lem:Generic1} it holds that $\cB_{z}(W)$ is a packing of $\cj$. Moreover, 
		$$\#\cB_z(W) \leq \chi(G_{\cj})+|L_{\cj}|+\frac{3}{2} \cdot s(M_{\cj})+\frac{4}{3} \cdot s(S_{\cj}) \leq 2+0+0+\frac{4}{3} \cdot r_z \leq \frac{4}{3} \cdot |T_{z}(W)|+9.$$ The first inequality holds by Lemma~\ref{lem:Generic1}. The second inequality holds by \eqref{eq:binTypes} and since $G_{\ci}$ is bipartite. 
		
	\qed \end{proof}
	Let $z = \{0,1\}$ and $T_z(W) = (A_1, \ldots,A_{r_z})$. For $i \in [r_z]$, let $\textsf{first}_i$ be the first subset of items in $W_V \cap A_i$ that has total size at least $\frac{1}{6}$, where the items are taken by some fixed non-increasing order of sizes; if there is no such subset then define $\textsf{first}_i = \emptyset$.

	\begin{myclaim}
		\label{claim:auxp6}
		Let $z = \{0,1\}$ and $i \in [r_z]$. Then, $s(\textnormal{\textsf{first}}_i) \leq \frac{1}{3}$; also, if $s(A_i \cap W_V) < \frac{1}{2}$ then $s(\textnormal{\textsf{first}}_i) \neq \emptyset$. 
	\end{myclaim}
	\begin{proof}
		We first show that $s(\textsf{first}_i) \leq \frac{1}{3}$. If there is an item larger than $\frac{1}{6}$ in $A_i \cap W_V$, it holds that $s(\textsf{first}_i) \leq \frac{1}{3}$ because $\textsf{first}_i$ contains a single item by definition. Otherwise, all items in $\textsf{first}_i$ are of sizes smaller than $\frac{1}{6}$; in this case, let $\ell = \argmin_{\ell' \in \textsf{first}_i} s(\ell)$. It holds that $s(\textsf{first}_i \setminus \{\ell\}) \leq \frac{1}{6}$ by the definition of $\textsf{first}_i$; hence, together with $\ell$, it holds that $s(\textsf{first}_i) \leq \frac{1}{6}+s(\ell) \leq \frac{1}{6} +\frac{1}{6} = \frac{1}{3}$. Now, assume that $s(A_i \cap W_V) > \frac{1}{2}$. Then, there must be a subset of $A_i \cap W_V$ of total size at least $\frac{1}{6}$. 
	\qed \end{proof}

	Now, define 
	
	\begin{equation}
		\label{eq:p6}
		\begin{aligned}
			\cB_{\{0,1\}}(W) ={} & \left((A_i \cup A_{i+1}) \cap (W_V \setminus (\textsf{first}_i \cup \textsf{first}_{i+1}))~|~i \in \{1,3, \ldots, r_z-1\}\right) \\ \oplus{} &\left((A_i \cup A_{i+1}) \cap \bar{W}_V~|~i \in \{1,3, \ldots, r_z-1\}\right) \\ \oplus{} & \left( \textsf{first}_i \cup \textsf{first}_{i+1} \cup \textsf{first}_{i+2}  ~|~i \in \{1,4, \ldots, r_z-2\}\right)
		\end{aligned}
	\end{equation}

\begin{myclaim}
		\label{claim:p6}
		Let $z = \{0,1\}$. Then, $\cB_{z}(W)$ is a packing of $\textnormal{\textsf{items}}(T_z(W))$ such that
		
		 $$\#	\cB_{z}(W) \leq \frac{4}{3} \cdot |T_{z}(W)|+7.$$
	
\end{myclaim}
	\begin{proof}
		
		For all $i \in \{1,3, \ldots, r_z-1\}$, $j \in  \{1,3, \ldots, r_z-1\}$, and $k \in \{1,4, \ldots, r_z-2\}$ it holds that $(A_i \cup A_{i+1}) \cap (W_V \setminus (\textsf{first}_i \cup \textsf{first}_{i+1}) \subseteq W_V$, $(A_j \cup A_{j+1}) \cap \bar{W}_V \subseteq \bar{W}_V$, and $ \textsf{first}_k \cup \textsf{first}_{k+1} \cup \textsf{first}_{k+2} \subseteq W_V$. Thus, $(A_i \cup A_{i+1}) \cap (W_V \setminus (\textsf{first}_i \cup \textsf{first}_{i+1})$, $(A_j \cup A_{j+1}) \cap \bar{W}_V$, and $\textsf{first}_k \cup \textsf{first}_{k+1} \cup \textsf{first}_{k+2}$ are independent sets in $G$.
		Moreover, it holds that
		
		 \begin{equation}
			\label{eq:np6}
			s((A_i \cup A_{i+1}) \cap (W_V \setminus (\textsf{first}_i \cup \textsf{first}_{i+1})) \leq 2 \cdot\frac{1}{2} = 1.
		\end{equation}
	
	 If $s(A_i \cap W_V) \leq \frac{1}{2}$ and $s(A_{i+1} \cap W_V) \leq \frac{1}{2}$ \eqref{eq:np6} follows. Otherwise, assume that $s(A_i \cap W_V) > \frac{1}{2}$; observe that $s(A_i \cap W_V) < \frac{2}{3}$ by \eqref{eq:binTypes}. Then, by Claim~\ref{claim:auxp6} it holds that $s(\textsf{first}_i) \geq \frac{1}{6}$ and it follows that $s((A_i) \cap (W_V \setminus \textsf{first}_i)) \leq \frac{1}{2}$; by symmetric arguments, if $s(A_{i+1} \cap W_V) > \frac{1}{2}$ it holds that $s((A_i) \cap (W_V \setminus \textsf{first}_{i+1})) \leq \frac{1}{2}$. By the above, \eqref{eq:np6} follows. 
		also, $s((A_{j} \cup A_{j+1}) \cap \bar{W}_V) \leq 2 \cdot \frac{1}{2}$ by \eqref{eq:binTypes}. Finally, $s(\textsf{first}_k \cup \textsf{first}_{k+1} \cup \textsf{first}_{k+2}) \leq 3 \cdot \frac{1}{3}$ by Claim~\ref{claim:auxp6}. By \eqref{eq:p6} it holds that $\#\cB_z(W) \leq \frac{r_z}{2}+\frac{r_z}{2}+\frac{r_z}{3} \leq  \frac{4}{3} \cdot r_z \leq \frac{4}{3} \cdot |T_{z}(W)|+7$.
		
	\qed \end{proof}

	Define
	\begin{equation}
		\label{eq:p7}
		\begin{aligned}
			\cB_{\{2,0\}}(W) ={} & \left(A_i \cap W_V~|~i \in [r_z]\right) \oplus \\
			{} & \left((A_i \cup A_{i+1} \cup A_{i+2}) \cap \bar{W}_V~|~i \in \{1,4, \ldots, r_z-2\}\right)
		\end{aligned}
	\end{equation}\begin{myclaim}
		\label{claim:p7}
		Let $z = \{2,0\}$. Then, $\cB_{z}(W)$ is a packing of $\textnormal{\textsf{items}}(T_z(W))$ such that $$\#	\cB_{z}(W) \leq \frac{4}{3} \cdot |T_{z}(W)|+7.$$
	\end{myclaim}
	\begin{proof}
		For all $i \in [r_z]$ and $j \in \{1,4, \ldots, r_z-2\}$ it holds that $A_i \cap W_V \subseteq W_V$, and $(A_j \cup A_{j+1}  \cup A_{j+2}) \cap \bar{W}_V  \subseteq \bar{W}_V$. Thus, $A_i \cap W_V$ and  $(A_j \cup A_{j+1}  \cup A_{j+2}) \cap \bar{W}_V$ are independent sets in $G$.
		Moreover, it holds that $s(A_i \cap W_V) \leq 1$ as $\OPT$ is a packing; also, $s((A_{j} \cup A_{j+1} \cup A_{j+2}) \cap \bar{W}_V) \leq 3 \cdot \frac{1}{3}$ because $s(A_i) - s(A_i \cap W_V) \leq 1-s(A_i \cap M_{\ci}) \leq 1-2 \cdot \frac{1}{3} = \frac{1}{3}$ by \eqref{eq:binTypes}. By \eqref{eq:p7} it holds that $\#\cB_z(W) \leq r_z+\frac{r_z}{3} \leq  \frac{4}{3} \cdot r_z \leq \frac{4}{3} \cdot |T_{z}(W)|+7$.
	\qed \end{proof}
	
	We combine the above sub-packings of $W_V$ into a packing of $W_V$; 
	specifically, 
	
	\begin{equation}
		\label{eq:p8}
		\cB(W) = \bigoplus_{z \in (Z\cup \{L\})} \cB_z(W).  
	\end{equation} 

By \eqref{eq:binTypes} it holds that $\bigcup_{z \in (Z\cup \{L\}} \{T_z(W)\}$ is a partition of the bins of $\OPT_W$. Hence, by Claims~\ref{claim:p1}-\ref{claim:p5} and 
\ref{claim:p6}-\ref{claim:p7}, as well as \eqref{eq:p8},  it holds that $\cB(W)$ is a packing of $\ci \cap W_V$. In addition, 

\begin{equation}
		\label{eq:BW}
		\begin{aligned}
			\#\cB(W) ={} & \sum_{z \in (Z\cup \{L\})} \#\cB_{z}(W) \\
			\leq{} & \frac{3}{2} \cdot |T_L(W)|+1+\sum_{z \in Z} \left(\frac{4}{3} \cdot |T_{z}(W)|+9\right) \\
			\leq{} & \frac{3}{2} \cdot |T_L(W)|+\frac{4}{3} \cdot \left(|\OPT_W|-|T_L(W)|\right)+O(1). 
		\end{aligned}
	\end{equation} 

The first equality holds by \eqref{eq:p8}. The first inequality holds 
by Claims~\ref{claim:p1}-\ref{claim:p5} and
\ref{claim:p6}-\ref{claim:p7}, as well as \eqref{eq:p8}. The last inequality holds by \eqref{eq:p8}. Hence, by \eqref{eq:BW} the packing $\cA_2$ of $W_V$ computed by Algorithm~\ref{alg:generic} contains
at most $\#\cB(W)+o(\OPT(\ci))$ bins.
Therefore, the number of bins in the packing $\cB$ returned by $\textsf{\Generic}(\ci)$ satisfies\footnote{We note that the last inequality in \eqref{eq:foo} holds asymptotically, i.e., for a sufficiently large $\OPT(\cI)$.}

	\begin{equation}
		\label{eq:foo}
		\begin{aligned}
			\#\cB \leq{} & \#\cB(X)+\#\cB(Y)+o(\OPT(\ci)) \\
			\leq{} & o(\OPT(\ci))+\sum_{W \in \{X,Y\}} \frac{3}{2} \cdot |T_L(W)|+\frac{4}{3} \cdot \left(|\OPT_W|-|T_L(W)|\right)+O(1) \\
			\leq{} & \frac{3}{2} \cdot |L_{\ci}|+\frac{4}{3} \cdot (\OPT(\ci)-|L_{\ci}|)+o(\OPT(\ci)). 
		\end{aligned}
	\end{equation} The second inequality holds by \eqref{eq:BW}. The third inequality holds by \eqref{eq:binTypes} and since each bin can contain at most one large item. 
\end{proof}

\noindent{\bf Proof of Lemma~\ref{lem:Generic}:} The proof follows by Lemma~\ref{lem:Generic1} and Lemma~\ref{lem:Generic2}. $\blacksquare$

	\hfill \break

\noindent{\bf Proof of Lemma~\ref{lem:ptas}:}
	Let $\ci =(V,E,w,\beta)$ be a BIS instance and $\eps>0$ be the error parameter; Also, let $A = \textsf{PTAS}(\ci =(V,E,w,\beta),\eps)$. We use several auxiliary claims.\begin{myclaim}
		\label{claim:ptasrunning}
		The running time of algorithm~\ref{alg:ptas} on $\ci,\eps$ is ${\textnormal{poly}(\ci)}^{O(\frac{1}{\eps})}$. 
	\end{myclaim}
	
	\begin{proof}
		Observe that the number of iterations of the {\bf for} loop in Step~\ref{step:forptas} is bounded by ${\textnormal{poly}(\ci)}^{O(\frac{1}{\eps})}$. Moreover, assuming that finding a maximum independent set can be computed in polynomial time, Step~\ref{step:ISptas} takes polynomial time. By the above, the proof follows.  
	\qed \end{proof}

	\begin{myclaim}
		\label{claim:ptasfeasible}
		$A$ is a feasible solution for $\ci$
	\end{myclaim}
	
	\begin{proof}
		
		If $A = \emptyset$ the proof follows. Otherwise, by Step~\ref{step:forptas}, Step~\ref{step:whileptas}, Step~\ref{step:updateSptas}, and Step~\ref{step:Aptas} there is an independent set $F \subseteq V$ in $G$, $ |F| \leq \eps^{-1}, w(F) \leq \beta$, and there is a maximum independent set $S'$ of $G_F$ and $w$ and $S \subseteq S'$ such that $A = F \cup S$ and $w(F \cup S) \leq \beta$. By Step~\ref{step:EF} it holds that $F \cup S'$ is an independent set in $G$ and therefore also $F \cup S$, because  $F \cup S \subseteq F \cup S'$.  
		
	\qed \end{proof}
	
	\begin{myclaim}
		\label{claim:ptasprofit}
		Let $\OPT$ be some optimal solution for $\ci$. Then, $w(A) \geq (1-\eps) \cdot w(\OPT)$. 
	\end{myclaim}
	
	\begin{proof}
		
		Let $L_{\OPT} = \{v \in \OPT~|~w(v) > \eps \cdot \beta\}$ be the set of {\em large} items of $\OPT$. Observe that $|L_{\OPT}| \leq \eps^{-1}, w(L_{\OPT}) \leq \beta$, and that $L_{\OPT}$ is an independent set in $(V,E)$; this is because $\OPT$ is a feasible solution of $\ci$.  By the {\bf for} loop in Step~\ref{step:forptas}, there is an iteration where $F = L_{\OPT}$. Also, observe that $\OPT \setminus L_{\OPT}$ is an independent set in $G_F$ by the definition of $G_F$. Therefore, in Step~\ref{step:ISptas} an independent set $S$ of $G_F$ w.r.t. $w$ is found such that $w(S) \geq w(L_{\OPT})$. If $w(F \cup S) \leq \beta$, then by Step~\ref{step:whileptas}, Step~\ref{step:ifptas}, and Step~\ref{step:Aptas}, it holds that $w(A) \geq w(F \cup S) \geq w(\OPT \setminus L_{\OPT})+w(L_{\OPT}) = w(\OPT)$. Otherwise, we have that 
		
		$$w(A) \geq w(F \cup S) \geq \beta - \eps \cdot \beta = (1-\eps) \cdot \beta \geq  (1-\eps) \cdot w(\OPT).$$ 
		The first inequality holds by Step~\ref{step:Aptas}. The second inequality holds by Step~\ref{step:whileptas}. The third inequality holds by the feasibility of $\OPT$.

	\qed \end{proof}
	
	The proof of Lemma~\ref{lem:ptas} follows by Claims~\ref{claim:ptasrunning}
	-\ref{claim:ptasprofit}. $\blacksquare$
	
	\hfill \break

\noindent{\bf Proof of Lemma~\ref{lem:SAP}:} By Definition~\ref{def:maxSAP}, the maximization problem of $\ci$ and $\cB$ is a special case of the {\em separable assignment problem (SAP)}~\cite{fleischer2011tight} defined as follows. Let $U$ be a set of $n$ bins, and $H$ 
a set of $m$ items; there is a value $f_{ij}$ for assigning item $j$ to bin $i$. We are also given a separate packing constraint
for each bin $i$. We assume that  if a set $I_i$ is feasible for a bin $i \in U$, then any subset of $I_i$ is also feasible for bin $i$. The goal is to find a feasible assignment of items to the bins which maximizes the aggregate value. For
each bin, we define the {\em single-bin subproblem} as the optimization problem over feasible sets associated with the
packing constraint for the bin.
 
We now show how the maximization problem of $\ci$ and $\cB$ can be cast as a SAP instance. Let $U = \cB$ be the set of bins and $H = I \setminus S$ the set of items; the value of assigning item $i$ to bin $j$ is $s(i)$.  For each bin $B_j$ for $j \in [t]$, the set of feasible subsets that can be packed in $B_j$ is $\cf_j = \{T \subseteq H~|~s(T) \leq 1-s(B_j), B_j \cup T \in \textsf{IS}(G_\ci)\}$. Note that by the definition of $\cf_j $ it holds that if some $T \in \cf_j$ then for all $T' \subseteq T$ it also holds that $T' \in \cf_j$. The objective is as in the maximization problem of $\ci$ and $\cB$, to find an assignment of items to bins of maximum aggregate value.
		
		By Definition~\ref{def:boundedIS} it holds that for each bin $B_j$ for $j \in [t]$ the single bin problem for $B_j$ is the BIS instance $(V_j,E_j,w_j,\beta_j)$, where the vertices are those which can create an independent set with the set $B_j$: $V_j = \{v \in H~|~\forall u \in B_j:(v,u) \notin E\}$, the edges are the induced edges on the set $V_j$: $E_j = \{(v,u)~|~v,u \in V_j\}$, the weight function is the size function of $\ci$: for all $v \in V_j: w_j(v) = s(v)$, and finally the budget is the remaining capacity of the bin $B_j$: $\beta_j = 1-s(B_j)$. By Definition~\ref{def:boundedIS} and the definition of SAP it holds that $\cj_j$ is the single bin sub-problem for $\ci$ and $B_j$. In addition, by Lemma~\ref{lem:ptas} there is a PTAS for the single bin sub-problem. Hence, by the results of \cite{fleischer2011tight}, there is an $(1-\frac{1}{e}-\eps)$-approximation for the maximization problem of $\ci$ and $\cB$
		$~~~~~~~~~~~~~~~~~~~~~~~~~~~~~~~~~~~~~~~~~~~~~~~~~~~~\blacksquare$

		\hfill \break

\noindent
{\bf Proof of Lemma~\ref{lem:MaxSolve}:} Let $\OPT = (O_1, \ldots, O_t)$ be some optimal packing of $\ci$. Let $K = \{v \in I \setminus L_{\ci}~|~v \in O_i, i \in [t], O_i \cap L_{\ci} \neq \emptyset\}$ be all non-large items that are packed in $\OPT$ in a bin that contain a large item. Then, 
			
			\begin{equation}
				\label{eq:x+y}
				\begin{aligned}
					s(I \setminus \textsf{items}(\cA)) 
					\leq{} & \OPT(\ci)-|L_{\ci}|+(\frac{1}{e} +\eps) \cdot s(K) \\
					\leq{} &  \OPT(\ci)-|L_{\ci}|+\left(\frac{1}{e}+\eps\right)\cdot\frac{|L_{\ci}|}{2}
				\end{aligned}
			\end{equation}
		
		 The first inequality holds since by Step~\ref{step:maxSolveSize}, Definition~\ref{def:maxSAP} and Lemma~\ref{lem:SAP} it holds that $s(\textsf{items}(\cA)) \geq s(L)+(1-\frac{1}{e} -\eps) \cdot s(K)$. The second inequality holds since $s(K) \leq \frac{|L_{\ci}|}{2}$ because each bin in $\OPT$ that contains an item from $L_{\ci}$ as at most half of its total size available for the items in $K$; moreover, there are  $|L_{\ci}|$ bins in which the items from $K$ are packed. Let $\cj = \ci \setminus  \textsf{items}(\cA)$. Therefore, by assigning $x = s(M_{\cj})$ and $y = s( S_{\cj})$, we have the following. 
		 
		  \begin{equation}
				\label{eq:x+y2}
				\begin{aligned}
					\#\cC ={} & \#\cA \oplus \cB \\ 
					={} & |L_{\ci}|+\# \textsf{\Generic}(\ci \setminus \textsf{items}(\cA)) \\
					\leq{} & |L_{\ci}|+\chi(G_{\cj})+|L_{\cj}|+\frac{3}{2} \cdot s(M_{\cj})+\frac{4}{3} \cdot s(S_{\cj}) \\
					\leq{} & |L_{\ci}|+\chi(G_{\ci})+\frac{3}{2} \cdot x+\frac{4}{3} \cdot y.
				\end{aligned}
			\end{equation} 
		
		The first equality follows by Step~\ref{step:maxSolveReturn}. The second equality follows by Lemma~\ref{lem:SAP}, Definition~\ref{def:maxSAP}, and Step~\ref{step:maxSolveGeneric}. The first inequality holds by Lemma~\ref{lem:Generic}. The last inequality holds by the definition  of $\cj, x$ and $y$. The proof follows by \eqref{eq:x+y}, \eqref{eq:x+y2}, and the definition of $x$ and $y$. $\blacksquare$

			\hfill \break

\noindent
{\bf Proof of Lemma~\ref{lem:Matching}:} Let $\OPT$ be some optimal packing of $\ci$; also, let $m$ be the number of bins in $\OPT$ that are packed with two items with $L_{\ci} \cup M_{\ci}$; note that there can be at most two items from  $L_{\ci} \cup M_{\ci}$ in a bin. Hence, all other items from $L_{\ci} \cup M_{\ci}$ are packed without any other item from $L_{\ci} \cup M_{\ci}$ in $\OPT$; it follows that $\OPT(\ci) \geq m+\left(|L_{\ci} \cup M_{\ci}| -2 \cdot m \right) = |L_{\ci} \cup M_{\ci}|-m$. Observe that $\cm$ found in Step~\ref{step:matchingMax} of Algorithm~\ref{alg:Matching} on $\ci$ contains at least $m$ edges by Definition~\ref{def:GraphMatching}. Therefore, the number of items from $L_{\ci} \cup M_{\ci}$ that are not in some edge in the matching $\cm$ is bounded by $|L_{\ci} \cup M_{\ci}| - 2 \cdot m$. Thus, the packing $\cB$ found in Step~\ref{step:matchingB} satisfies $\#\cB \leq |L_{\ci} \cup M_{\ci}| - 2 \cdot m+m \leq \OPT(\ci)$. Now,  let $\cj = \ci \setminus (L_{\ci} \cup M_{\ci}))$. To conclude the proof, 
					
					\begin{equation*}
					\begin{aligned}
						\#\cA ={} &  \#\cB \oplus \#\textsf{\Generic}(\ci \setminus (M_{\ci} \cup L_{\ci})) \\ 
						={} &  \#\cB+\#\textsf{\Generic}(\ci \setminus (M_{\ci} \cup L_{\ci}))\\
						\leq{} & \OPT(\ci)+\#\textsf{\Generic}(\ci \setminus (M_{\ci} \cup L_{\ci})) \\
						\leq{} & \OPT(\ci)+\chi(G_{\cj})+|L_{\cj}|+\frac{3}{2} \cdot s(M_{\cj})+\frac{4}{3} \cdot s(S_{\ci}) \\
						={} & \OPT(\ci)+\chi(G_{\ci})+\frac{4}{3} \cdot s(S_{\ci}).
					\end{aligned}
				\end{equation*}  
			
			he first equality holds by Step~\ref{step:matchingReturn}. The first inequality holds because $\#\cB \leq \OPT(\ci)$. The second inequality holds by Lemma~\ref{lem:Generic}. The last inequality holds by the definition of $\cj$. $\blacksquare$

			\hfill \break

\begin{lemma}
\label{lem:2+4/9}
Given a \textnormal{BPC} instance $\ci = (I,s,E)$, Algorithm \textnormal{\textsf{ApproxBPC}} returns in polynomial time in $|\ci|$ a packing $\cA$ of $\ci$ such that  $\#\cA \leq \left(2+\frac{4}{9}\right) \cdot \OPT(\ci)$. 
\end{lemma}

\begin{proof}

				We split the proof into two complementary cases. \begin{enumerate}
					\item $s(S_{\ci}) \leq \frac{s(I)}{3}$. Then, \begin{equation*}
						\label{eq:sB2}
						\begin{aligned}
							\#\textsf{Matching}(\ci) \leq{} & \OPT(\ci)+ \chi(G_{\ci})+\frac{4}{3} \cdot s(S_{\ci}) \\
							\leq{} & 2 \cdot  \OPT(\ci)+\frac{4}{3} \cdot  \frac{s(I)}{3} \\
							={} &  \left(2+\frac{4}{9}\right) \cdot \OPT(\ci).
						\end{aligned}
					\end{equation*} The first inequality holds by by Lemma~\ref{lem:Matching}. The second inequality holds since   $s(S_{\ci}) \leq \frac{s(I)}{3}$. Hence, the proof follows by the above and by Step~\ref{step:BPCA} and Step~\ref{step:BPCReturn} of Algorithm~\ref{alg:BPC}. 
					
					\item $s(S_{\ci}) > \frac{s(I)}{3}$. Then, by Lemma~\ref{lem:MaxSolve}, there are $0 \leq x\leq s(M_{\ci})$ and $0\leq y \leq s(S_{\ci})$ such that the following holds.  \begin{enumerate}
						\item $x+y \leq \OPT(\ci)-|L_{\ci}|+\left(\frac{1}{e}+\eps\right)\cdot\frac{|L_{\ci}|}{2}$.
						\item $\#\textsf{MaxSolve}(\ci,\eps) \leq \chi(G_{\ci})+|L_{\ci}|+\frac{3}{2} \cdot x+\frac{4}{3} \cdot y$.
					\end{enumerate} In this case, \begin{equation*}
						\begin{aligned}
							\#\textsf{MaxSolve}(\ci,\eps)  \leq{} &  \chi(G_{\ci})+|L_{\ci}|+\frac{3}{2} \cdot x+\frac{4}{3} \cdot y\\
							\leq{} & \chi(G_{\ci})+|L_{\ci}|+\frac{1}{6} \cdot x+\frac{4}{3} \cdot (x+y) \\
							\leq{} & \chi(G_{\ci})+|L_{\ci}|+\frac{1}{6} \cdot x+\frac{4}{3} \cdot \left(\OPT(\ci)-|L_{\ci}|+\left(\frac{1}{e}+\eps\right)\cdot\frac{|L_{\ci}|}{2} \right) \\
							\leq{} &  \chi(G_{\ci})+\frac{x}{6} + \frac{4}{3} \cdot \OPT(\ci) \\
							\leq{} &   \chi(G_{\ci})+\frac{s(M_{\ci})}{6} + \frac{4}{3} \cdot \OPT(\ci)\\
							\leq{} &   \chi(G_{\ci})+\frac{s(I)-s(S_{\ci}))}{6} + \frac{4}{3} \cdot \OPT(\ci)\\
							\leq{} &   \chi(G_{\ci})+\frac{\frac{2 \cdot \OPT(\ci)}{3}}{6} + \frac{4}{3} \cdot \OPT(\ci)\\
							={} &  \left(2+\frac{4}{9}\right) \cdot \OPT(\ci)\\
								\leq{} &  2.445 \cdot \OPT(\ci)
						\end{aligned}
					\end{equation*} The first inequality holds by (b). The third inequality holds by (a). The fourth inequality holds for all $0<\eps<0.1$. The second inequality from the end holds since $s(S_{\ci}) > \frac{s(I)}{3}$. Hence, the proof follows by the above and by Step~\ref{step:BPCA} and Step~\ref{step:BPCReturn} of Algorithm~\ref{alg:BPC}. $\blacksquare$
					
				\end{enumerate}
				
					\end{proof}

				\begin{lemma}
					\label{lem:asym}
					Given a \textnormal{BPB} instance $\ci = (I,s,E)$, Algorithm \textnormal{\textsf{ApproxBPC}} returns in polynomial time in $|\ci|$ a packing $\cA$ of $\ci$ such that  $\#\cA \leq 1.391 \cdot \OPT(\ci)+o(\OPT(\ci))$. 
				\end{lemma}

		\begin{proof}

				First, by Lemma~\ref{lem:Generic}, Step~\ref{step:BPCA}, and Step~\ref{step:BPCReturn} of Algorithm~\ref{alg:BPC}, it holds that \begin{equation}
					\label{eq:asym1f}
					\#\cA \leq \frac{3}{2} \cdot |L_{\ci}|+\frac{4}{3} \cdot \left(\OPT(\ci)-|L_{\ci}|\right)+o(\OPT(\ci)).
				\end{equation} Second, by Lemma~\ref{lem:MaxSolve}, there are $0 \leq x\leq s(M_{\ci})$ and $0\leq y \leq s(S_{\ci})$ such that the following holds.  \begin{enumerate}
					\item $x+y \leq \OPT(\ci)-|L_{\ci}|+\left(\frac{1}{e}+\eps\right)\cdot\frac{|L_{\ci}|}{2}$.\label{cond:1}
					\item $\#\textsf{MaxSolve}(\ci,\eps) \leq |L_{\ci}|+\frac{3}{2} \cdot x+\frac{4}{3} \cdot y+o(\OPT(\ci))$.\label{cond:2}
				\end{enumerate} Using the above,

					\begin{equation}
						\label{eq:maxsolvef}
						\begin{aligned}
							\#\textnormal{\textsf{ApproxBPC}}(\ci) 
							&	\leq{} 	\#\textsf{MaxSolve}(\ci,\eps)  \\ 
							&	\leq{}   o(\OPT(\ci))+|L_{\ci}|+\frac{3}{2} \cdot x+\frac{4}{3} \cdot y\\
							&	\leq{}   o(\OPT(\ci))+|L_{\ci}|+\frac{3}{2} \cdot (x+y)\\
							&	\leq{}  o(\OPT(\ci))+|L_{\ci}|+\frac{3}{2} \cdot \left(\OPT(\ci)-|L_{\ci}|+\left(\frac{1}{e}+\eps\right) \cdot\frac{|L_{\ci}|}{2} \right) \\ 
							&	\leq{}  o(\OPT(\ci))-0.316\cdot|L_{\ci}|+\frac{3}{2} \cdot \OPT(\ci). 		
						\end{aligned}
					\end{equation} The first inequality holds by by Step~\ref{step:BPCA} and Step~\ref{step:BPCReturn} of Algorithm~\ref{alg:BPC}. The second inequality holds by Condition~\ref{cond:2}. The third inequality holds by Condition~\ref{cond:1}. The fourth inequality holds for all $0<\eps<0.0001$.
					To conclude, we split the proof into two cases. \begin{itemize}
						
						\item $|L_{\ci}| < 0.345296 \cdot \OPT(\ci)+o(\OPT(\ci))$.Thus, it holds that $\#\textnormal{\textsf{ApproxBPC}}(\ci) \leq 1.391 \cdot \OPT(\ci)+o(\OPT(\ci))$ by \eqref{eq:asym1f}. 
						
						\item $|L_{\ci}| \geq 0.345296 \cdot \OPT(\ci)$. Thus, it holds that $\#\textnormal{\textsf{ApproxBPC}}(\ci) \leq 1.391 \cdot \OPT(\ci)$ by \eqref{eq:maxsolvef}. $\blacksquare$
						
					\end{itemize}
						\end{proof}
					\noindent
						{\bf Proof of Theorem~\ref{thm:BPC}:} The proof follows by Lemma~\ref{lem:2+4/9} and Lemma~\ref{lem:asym}. $\blacksquare$

\section{Omitted Proofs from Section~\ref{sec:split}}
\label{sec:proofsSplit}

Towards proving Lemma~\ref{lem:KP2}, we now present an FPTAS for BIS on a split graph. A simple observation is that at most one vertex from $K_{G}$ can be in the solution. Therefore, an FPTAS is obtained by iterating over all vertices in $K_{G}$, as well as the case where the solution contains no vertex from this set.
 Consider the iteration in which the bin contains
 a single vertex (or, no vertex) from $K_{G}$ as in an optimal solution. Then, additional vertices not adjacent to this vertex are added using an FPTAS for the classic {\em knapsack} problem. An input for knapsack is $\ci = (I,p,c,\beta)$, where $I$ is a set of items, and $p,c: I \rightarrow \mathbb{R}_{\geq 0}$ are a profit function and a cost function, respectively; also, $\beta \in \mathbb{R}_{\geq 0}$ is a budget. A solution for $\ci$ is $S \subseteq I$ such that $c(S) \leq \beta$; the goal is to find a solution of $\ci$ of maximum profit. We use the next well known result for knapsack (see, e.g., \cite{vazirani,deng2023approximating}).  
 \begin{lemma}
	\label{lem:KP}
	There is an algorithm \textnormal{\textsf{KP}} that is an \textnormal{FPTAS} for the $0/1$-knapsack problem. 
\end{lemma} We use $N(v) = \{u \in I~|~ (u,v) \in G_{\ci}\}$ to denote the set of {\em neighbors} of some $v \in I$. 
Assume w.l.o.g. that $I = \{1,\ldots, n\}$ for some $n \in \mathbb{N}$; then it follows that $0 \notin I$. The pseudocode of our FPTAS for BIS on split graphs is given in Algorithm~\ref{alg:FPTAS}.

\begin{algorithm}[h]
	\caption{$\textsf{FPTAS-BIS}(G = (V,E), w, \beta, \eps)$}
	\label{alg:FPTAS}
	
		\begin{algorithmic}[1]
	
	\For{$v \in K_{G}$}
		
		\State{$\cA_v \leftarrow \textnormal{\textsf{KP}}((S_G \setminus N(v),w,w,\beta-w(v)),\eps)$.\label{step:kp:v}}
		
	\EndFor
	
	\State{$\cA_0 \leftarrow \textnormal{\textsf{KP}}((S_{G},w,w,\beta),\eps)$. \label{step:kp:0}}
	
	\State{Return $\argmax_{\cA \in \{\cA_v | v \in K_{G} \cup \{0\}} s(\cA)$.\label{step:FPTAS:Return}}
		\end{algorithmic}
\end{algorithm}

\noindent {\bf Proof of Lemma~\ref{lem:KP2}:} Consider a BIS instance $\ci$ composed of a split graph $G = (V,E)$, a weight function $w:V \rightarrow \mathbb{R}_{\geq 0}$, a budget $\beta \in \mathbb{R}_{\geq 0}$, and an error parameter $\eps>0$. Also, let $\OPT$ be an optimal solution for $\ci$. By Definition~\ref{def:boundedIS}, $\OPT$ is an independent set in $G$; thus, $|\OPT \cap K_{G}| \leq 1$. If $|\OPT \cap K_{G}| = 0$ then $\cA_0$ is a $(1-\eps)$-approximation for $\ci$ by Lemma~\ref{lem:KP}, since $\OPT \subseteq S_{G}$ (Algorithm \textsf{KP} in particular must return a $(1-\eps)$-approximation for $\OPT$). Otherwise, $|\OPT \cap K_{G}| = 1$; then, there is a single $v \in K_G$ such that $\OPT \cap K_G = \{v\}$. Therefore, $\OPT \setminus \{v\} \subseteq S_G \setminus N(v)$ (as $\OPT$ is an independent set in $G$). Then, $\cA_v$ is a $(1-\eps)$-approximation for $\ci$ by Lemma~\ref{lem:KP} (Algorithm \textsf{KP} in particular must return a $(1-\eps)$-approximation for $\OPT \setminus \{v\}$). Overall, Algorithm \textsf{FPTAS-BIS} is a $(1-\eps)$-approximation for $\ci$. It follows that Algorithm \textsf{FPTAS-BIS} is an FPTAS (for split graphs) since the running time is $|K_G| \cdot \textnormal{poly}\left(|\ci|, \frac{1}{\eps}\right) =  \textnormal{poly}\left(|\ci|, \frac{1}{\eps}\right)$ by Lemma~\ref{lem:KP}. $\blacksquare$

~\\
\noindent {\bf Proof of Theorem~\ref{thm:split}:} Let $\ci = (I,s,E)$ be a BPS instance. Trivially, the algorithm returns a packing of $\ci$; we are left to show the approximation guarantee. Note that $K_{G_{\ci}}$ is a clique in $G_{\ci}$; thus, the items in $K_{G_{\ci}}$ must be packed in different bins and it follows that $\OPT(\ci) \geq |K_{G_{\ci}}|$. Moreover, it holds that

 \begin{equation}
	\label{eq:boundings}
	2 \leq \OPT(\ci) \leq \textsf{\Generic}(\ci) \leq \chi(G_{\ci})+\ceil{2 \cdot s(I)} \leq |K_{G_{\ci}}|+1+\ceil{2 \cdot s(I)}. 
\end{equation}  

The third inequality holds using Lemma~\ref{lem:Generic}. Therefore, by \eqref{eq:boundings} there is  $\alpha^* \in \left\{0,1,\ldots, \ceil{2 \cdot s(I)} +1\right\}$ such that $\alpha^*+|K_{G_{\ci}}| = \OPT(\ci)$. Then, by Lemma~\ref{lem:KP2} and Lemma~\ref{lem:SAP}, Algorithm \textsf{MaxSize} is a $\left( 1-\frac{1}{e}\right)$-approximation for the maximization problem of $\cI$ and the initial partial packing $\cB_{\alpha^*}$. Thus, 

\begin{equation*}
	\label{eq:sB2}
	\begin{aligned}
		\#\cA^*_{\alpha^*} ={} & \#\cA_{\alpha^*}+\#\textsf{FFD}(\ci \setminus \textsf{items}(\cA_{\alpha^*})) \\
		\leq{} & \OPT(\ci)+\frac{2 \cdot \OPT(\ci)}{e} \\
		={} & \left(1+\frac{2}{e}\right) \cdot \OPT(\ci). 
	\end{aligned}
\end{equation*} The inequality holds by the selection of $\alpha^*$ and since $s(\textsf{items}(\cA_{\alpha^*})) \geq \left( 1-\frac{1}{e}\right)s(I)$; therefore, it follows that the total size of items in the instance $\ci \setminus \textsf{items}(\cA_{\alpha^*})$ is bounded by $\frac{s(I)}{e}$. If $\frac{s(I)}{e} \leq 1$ then $\#\textsf{FFD}(\ci \setminus \textsf{items}(\cA_{\alpha^*})) = 1$ and otherwise it holds that $\#\textsf{FFD}(\ci \setminus \textsf{items}(\cA_{\alpha^*})) \leq 2 \cdot \frac{s(I)}{e}$ by a simple bound on the well known approximation guarantee of \textsf{FFD}. Thus, the residual instance $\ci \setminus \textsf{items}(\cA_{\alpha^*})$ is packed via \textsf{FFD} in at most $\max\{2 \cdot \frac{s(I)}{e}, 1\} \leq \frac{2 \cdot \OPT(\ci)}{e}$ bins (recall that $\OPT(\ci) \geq 2$).  $\blacksquare$
\section{Omitted Proofs from Section~\ref{sec:HardnessBP}}
\label{sec:proofsBPHardness}

\begin{lemma}
	\label{1}
	Given some $k > 4$ and a $k$-\textnormal{restricted B3DM} instance $\cj = (X,Y,Z,T)$, it holds that $\OPT(\cj) \geq \frac{|X|+|Y|+|Z|+|T|}{4 \cdot c^3}$. 
\end{lemma}
\begin{proof}
	
	Assume towards contradiction that $\OPT(\cj) < p$. Let $\OPT \subseteq T$ be some optimal solution for $\cj$. Also, for $t  \in \OPT$ such that $t = (x,y,z)$, we say that $x,y,z \in t$ for simplicity.  
	
	\begin{equation}
		\label{eq:pwp}
		\begin{aligned}
			&	|\{(x,y,z) \in T \setminus \OPT ~|~ \exists t \in \OPT \text{ s.t. } x \in t \text{ or } y \in t \text{ or } z \in t\}|\\ \leq{} & \sum_{(x,y,z) \in \OPT~} \sum_{w \in \{x,y,z\}} 	|\{t' \in T \setminus \OPT ~|~ w \in t', \exists t \in \OPT \text{ s.t. } w \in t\}| \\
			\leq{} & 3 \cdot \OPT(\cj) \cdot c < 3 \cdot \frac{|T|}{c^3} \cdot c \leq \frac{|T|}{c}
		\end{aligned}
	\end{equation} 

The first inequality holds by the union bound. The second inequality holds since $\cj$ is a Bounded-3DM instance. 

\begin{equation}
		\label{eq:pwp2}
		\begin{aligned}
			& |\{ t \in T \setminus \OPT~|~\OPT \cup \{t\} \text{ is a feasible solution for } \cj\}| \\
			\geq{} &	|T|-\OPT(\cj)-|\{(x,y,z) \in T \setminus \OPT ~|~ \exists t \in \OPT \text{ s.t. } x \in t \text{ or } y \in t \text{ or } z \in t\}|\\ 
			\geq{} & |T|-\OPT(\cj) - \frac{|T|}{c} \\
			\geq{} & |T| \cdot \left(1- \frac{2}{c}\right) \\
			\geq{} & 4 \cdot  \frac{2}{3}\\
			\geq{} & 1.\\
		\end{aligned}
	\end{equation} 

The first inequality holds by the definition of solution of a B3DM instance. The second inequality holds by \eqref{eq:pwp}. The third inequality holds because $\OPT(\cj) < \frac{|T|}{c^3}$ by our assumption. The fourth inequality holds because $\cj$ is $k$-restricted for $k > 4$. By \eqref{eq:pwp2}, there is an elements $t \in T$ such that $\OPT \cup \{t\}$ is a solution of $\cj$. This is a contradiction to the optimality of $\OPT$. Thus, $\OPT(\cj) \geq \frac{|T|}{c^3}$. Therefore, 

$$\OPT(\cj)  \geq \frac{|T|}{c^3}  \geq \frac{|T|}{4 \cdot c^3}+\frac{|T|}{4 \cdot c^3}+\frac{|T|}{4 \cdot c^3}+\frac{|T|}{4 \cdot c^3} \geq \frac{|T|}{4 \cdot c^3}+\frac{|X|}{4 \cdot c^3}+\frac{|Y|}{4 \cdot c^3}+\frac{|Z|}{4 \cdot c^3}.$$ 

The last inequality holds because each $u \in X \cup Y \cup Z$ appears in at least one element of $T$. $\blacksquare$
\end{proof}

\noindent {\bf Proof of Lemma~\ref{1-rest}:}
Assume towards contradiction that for all $\alpha>1$ there is an algorithm \textsf{Restricted-Solver} that is an $\alpha$-approximation for $k$-restricted \textnormal{B3DM} problem, for some fixed $k \in \mathbb{N}$. Let $\textnormal{\textsf{Enum}}_k$ be the algorithm that given a B3DM instance $\cj = (X,Y,Z,T)$ iterates over all subsets of $T$ of at most $k$ elements and chooses the best feasible solution found in the enumeration.  
\begin{observation}
		\label{ob:restricted}
		Given a non-k-restricted \textnormal{B3DM} instance $\cj$, Algorithm $\textnormal{\textsf{Enum}}_k$ returns an optimal solution for $\cj$. 
\end{observation}
	
	Now, given a B3DM instance $\cj = (X,Y,Z,T)$, we compute $\textnormal{\textsf{Enum}}_k(\cj)$ and $\textnormal{\textsf{Restricted-Solver}}_k(\cj)$ and return the better solution among the solutions found by the two algorithms. Note that the running time is polynomial by our assumption for \textnormal{\textsf{Restricted-Solver}} and since $k$ is a constant. Indeed, the number of subsets of $T$ with at most $k$ elements is $O(|T|^k)$. If $\cj$ is a restricted B3DM instance then we return a solution of cardinality at least $\frac{\OPT(\cj)}{\alpha}$; otherwise, by Observation~\ref{ob:restricted} we return an optimal solution. Overall, we conclude that the above is an $\alpha$-approximation for B3DM for any $\alpha>1$. A contradiction aince B3DM does not admit a PTAS \cite{kann1991maximum}. $\blacksquare$

~\\
	
\noindent {\bf Proof of Theorem~\ref{thm:BPH}:}	We give the proof for BPB. At the end of the proof, we explain how a minor modification in the reduction can be used to obtain a similar result also for BPS. Assume towards contradiction that there is a polynomial-time algorithm \textsf{ALG} such that \textsf{ALG} is an asymptotic $d$-approximation for BPB for all $d > 1$. By Lemma~\ref{1-rest} there is $\alpha>1$ such that there is no $\alpha$-approximation for the $k$-restricted \textnormal{B3DM} problem, for all $k \in \mathbb{N}$. Therefore, there is $n \in \mathbb{N}$ such that for any BPB instance $\mathcal{I}$ with $\OPT(\mathcal{I}) \geq n$ it holds that $\#\text{\textsf{ALG}}(\mathcal{I}) < d \cdot  \OPT(\mathcal{I})$ for 

\begin{equation}
		\label{eq:d}
		d = \frac{(1-\frac{1}{\alpha})}{80 \cdot c^3} 
	\end{equation} 

The selection of $d$ becomes clear towards the end of the proof. We show that using \textsf{ALG} we can approximate  $n$-restricted B3DM within a constant ratio of less than  $\alpha$; this is a contradiction to Lemma~\ref{1-rest}. For completeness of this section, we repeat some of the notations given in Section~\ref{sec:HardnessBP}. Let $\cj = (X,Y,Z,T)$ be an $n$-restricted instance of B3DM and let

\begin{equation}
	\label{eq:U}
	\begin{aligned}
	U ={} & X \cup Y \cup Z, 
	\\
	X ={} & \{x_1, \ldots, x_{\tilde{x}}\}, Y = \{y_1, \ldots, y_{\tilde{y}}\}, Z = \{z_1, \ldots, z_{\tilde{z}}\}, T = \{t_1, \ldots, t_{\tilde{t}}\}. 
	\end{aligned}
\end{equation}
  We assume w.l.o.g. that each $u \in X \cup Y \cup Z$ appears in at least one triple of $T$ (otherwise this element can be omitted from the instance without changing the set of solutions).
We define additional {\em filler} elements; these elements are packed in the optimum of our constructed BPB instance together with elements not taken to the solution for $\cj$. Since we do not know the exact value of the optimum of $\cj$, except that it is a number between $n$ and $|T|$, we define a family of instances with different number of filler items. Specifically, for all $i \in \{n,n+1, \ldots, |T|\}$, let $P_i,Q_i$ be a set of $\tilde{t}-i$ elements and a set of $\tilde{x}+\tilde{y}+\tilde{z}-3 \cdot i$ elements, respectively, such that $P_i \cap U = \emptyset$ and $Q_i \cap U = \emptyset$. Let $I_i = U \cup P_i \cup T \cup Q$ and define the following bipartite graph: $G_i = (I_i,E_i)$, where $E = E_X \cup E_Y \cup E_Z$ such that the following holds.

	\begin{equation}
		\label{eq:edges}
		\begin{aligned}
			E_X ={} & \{ (x,t)~|~ x \in X, t = (x',y,z) \in T, x \neq x' \}\\
			E_Y ={} & \{ (y,t)~|~ y \in Y, t = (x,y',z) \in T, y \neq y' \}
			\\
			E_Z ={} & \{ (z,t)~|~ z \in Z, t = (x,y,z') \in T, z \neq z' \}
			\\
		\end{aligned}
	\end{equation}

	In words, we define the bipartite graph between the items in $U$ and the triples in $T$ and connect edges between pairs of an item and  a triple such that the item does not belong to the triple. 
	Now, define the following BPB instances (for all $i \in \{n,n+1, \ldots, |T|\}$): 
	
	\begin{equation}
		\label{eq:insta}
		\begin{aligned}
			{} & \mathcal{I}_i = (I_i,s,E) \text{ s.t. }\\
			{} & \forall u \in U, p \in P_i, q \in Q_i, t \in T: \\
			{} & s(u) = 0.15, s(p) = 0.45, s(q) = 0.85, s(t) = 0.55.
		\end{aligned}
	\end{equation}

	By \eqref{eq:insta} and that $\cj$ is $n$-restricted,then in particular, there are at least $|T| =  \tilde{t} \geq n$ number of items with sizes strictly larger than $\frac{1}{2}$ in $\ci_i$. Thus, $\OPT(\ci_i) \geq n$. Therefore, we conclude that \textsf{ALG} can find a feasible packing of $\mathcal{I}_i$ in at most $d  \cdot \OPT(\mathcal{I}_i)$ bins. Let $ \textsf{ALG}(\ci_i) = (B_1, \ldots,B_{R_i})$ for $R_i \leq d \cdot \OPT(\ci_i)$; when understood from the context, we use $R = R_i$. Now, define the following solution for $\cj$: 
	
	\begin{equation}
		\label{eq:solJ}
		S_i = \{(x,y,z) \in T~|~\exists r \in [R] \text{ s.t. } B_r = \{x,y,z,(x,y,z)\}\}. 
	\end{equation} 
	\begin{myclaim}
		\label{claim:solJ}
		$S_i$ is a solution of $\cj$. 
	\end{myclaim}
	\begin{proof}
		First, by \eqref{eq:solJ} it follows that $S_i \subseteq T$. Now, let $(x,y,z) \in S_i$. Since $(B_1, \ldots,B_R)$ is a partition of $I_i$ by the definition of packing, $x,y,z$ cannot appear in any other element in $S_i$. 
	\qed \end{proof}
	
	\begin{myclaim}
		\label{claim:siOPT}
		For $i = \OPT(\cj)$ it holds that $\OPT(\ci_i) = s(I_i) = \tilde{t}+\tilde{x}+\tilde{y}+\tilde{z}-3 \cdot i$. 
	\end{myclaim}
	\begin{proof}
		Let $\OPT$ be an optimal solution for $\cj$. Let $P_i = \{p_1, \ldots, p_{\tilde{t}-i}\}$ and $Q_i = \{q_1, \ldots, q_{\tilde{x}+\tilde{y}+\tilde{z}-3 \cdot i}\}$. Also, let $X_i = \{x'_{1}, \ldots, x_{\tilde{x}-i}\}$, $Y_i = \{y'_{1}, \ldots, y'_{\tilde{y}-i}\}$, and $Z_i = \{z'_{1}, \ldots, z'_{\tilde{z}-i}\}$ be all elements in $X,Y,Z$ that do not belong to any element in $\OPT$, respectively; finally, let $T \setminus \OPT = \{t_1, \ldots, t_{\tilde{t}-i}\}$. Now, define the following packing for $\ci_i$. 
		
		\begin{equation}
			\label{eq:packk}
			\begin{aligned}
				{} & W_i = A\oplus B \oplus C \text{ s.t. }\\
				{} & A = \left(\{x,y,z,(x,y,z)\}~|~(x,y,z) \in \OPT\right)\\
				{} & B = \left(\{t_k,p_k\}~|~k \in [\tilde{t}-i]\right)\\
				{} & C = \left(\{x'_k,q_k\}~|~k \in \{1,\ldots,\tilde{x}-i\}\right) \\
				\oplus{} &
				\left(\{y'_k,q_k\}~|~k \in \{\tilde{x}-i+1, \ldots, \tilde{x}+\tilde{y}-2 \cdot i\}\right) \\
				\oplus{} &  \left(\{z'_k,q_k\}~|~k \in \{\tilde{x}+\tilde{y}-2 \cdot i+1, \ldots, \tilde{x}+\tilde{y}+\tilde{z}-3 \cdot i\} \right) 
			\end{aligned}
		\end{equation} 
	
	By \eqref{eq:edges}, \eqref{eq:insta}, and \eqref{eq:packk} it holds that $W_i$ is a packing of $\ci_i$ such that $W_i = (R_1, \ldots,R_r)$ and for all $k \in [r]$ it holds that $s(R_k) = 1$ and therefore $\# W_i = s(I_i)$. The proof follows since $\OPT(\ci_i) \geq s(I_i)$ and thus $W_i$ is an optimal packing for $\ci_i$.  Moreover, by \eqref{eq:packk} it holds that $r = \tilde{t}+\tilde{x}+\tilde{y}+\tilde{z}-3 \cdot i$. 
	\qed \end{proof}
	
	\begin{myclaim}
		\label{claim:Sfinal}
		$|S_i| \geq \frac{\OPT(\cj)}{\alpha}$. 
	\end{myclaim}
	\begin{proof}
		Below we give a lower bound for the number of bins in $ \textsf{ALG}(\ci_i) = (B_1, \ldots,B_R)$ that are packed with total size at most $0.9$. 
		
		\begin{equation}
			\label{eq:iSi1}
			\begin{aligned}
				|\{r \in [R]~|~s(B_r) \leq 0.9\}| = {} & R - |\{r \in [R]~|~s(B_r) > 0.9\}| \\
				\geq{} &R - |S_i| - |Q_i|-|P_i|\\
				\geq{} &\tilde{t}+\tilde{x}+\tilde{y}+\tilde{z}-3 \cdot i - |S_i| -(\tilde{x}+\tilde{y}+\tilde{z}-3 \cdot i)-(\tilde{t}-i) \\ 			
				={} & i-|S_i|
			\end{aligned}
		\end{equation} 
	
	The first inequality holds by the sizes in \eqref{eq:insta} and \eqref{eq:solJ}; that is, by \eqref{eq:insta} all bins in $ \textsf{ALG}(\ci_i)$ that do not contain an item from $Q_i$, an item from $P_i$, or belong to $S_i$ cannot be of total size more than $0.9$. The second inequality holds by Claim~\ref{claim:siOPT}. 
		We use the following inequality, 
		
		\begin{equation}
			\label{eq:iSi2}
			\begin{aligned}
				\OPT(\cj) \geq \frac{\tilde{t}+\tilde{x}+\tilde{y}+\tilde{z}}{4 \cdot c^3} \geq \frac{\tilde{t}}{4 \cdot c^3} \geq \frac{\OPT(\ci_i)}{4 \cdot c^3}. 
			\end{aligned}
		\end{equation} 
	
	The first inequality holds by Lemma~\ref{1}. The last inequality holds by \eqref{eq:insta}. Let $F = 	|\{r \in [R]~|~s(B_r) \leq 0.9\}|$ be the number of {\em bad} bins in $\textsf{ALG}(\ci_i)$. Therefore,

		\begin{equation}
			\label{eq:iSi3}
			\begin{aligned}
				F \geq i -|S_i| \geq \OPT(\cj)-\frac{\OPT(\cj)}{\alpha} \geq \frac{(1-\frac{1}{\alpha})}{4 \cdot c^3} \cdot \OPT(\ci_i). 
			\end{aligned}
		\end{equation} 
	
	The first inequality holds by \eqref{eq:iSi1}. The second inequality holds by the assumption on $|S_i|$. The last inequality holds by ~\eqref{eq:iSi2}. Thus, 
	
	\begin{equation}
			\label{eq:f14}
			R \geq F+(s(I_i)-\frac{9}{10} \cdot F) = \frac{F}{10}+\OPT(\ci) \geq \frac{(1-\frac{1}{\alpha})}{10 \cdot 4 \cdot c^3}  \cdot \OPT(\ci)+\OPT(\ci).
		\end{equation} 
	
	The first inequality holds because the $R$ bins in $\textsf{ALG}(\ci_i)$ must contain the total size of $I_i$; thus, since each bin in the bad bins is at most $\frac{9}{10}$ full, the number of non-bad bins is at least the total size remaining by deducting the upper bound of total size $\frac{9}{10} \cdot F$ on the total size packed in the bad bins. The first equality holds by Claim~\ref{claim:siOPT}. The last inequality holds by \eqref{eq:iSi3}. Finally, by \eqref{eq:f14} we reach a contradiction since $R \leq d \cdot \OPT(\ci_i) < \frac{(1-\frac{1}{\alpha})}{40 \cdot c^3}  \cdot \OPT(\ci)+\OPT(\ci)$ by \eqref{eq:d}. 
	\qed \end{proof}
	
	Finally, by Claim~\ref{claim:Sfinal} we reach a contradiction to the existence of $\textsf{ALG}$. By iterating over all $i' \in \{n, \ldots, |T|\}$, constructing $\ci_{i'}$, and returning the best solution $S_{i^*}$, we can in particular find $i = \OPT(\cj)$ and return $S_i$; by Claim~\ref{claim:Sfinal} it holds that the returned solution is an $\alpha$-approximation. The running time of constructing the reduction is polynomial, since $|T|$ is polynomial, the constructions in \eqref{eq:insta} and \eqref{eq:solJ} are polynomial, and $\textsf{ALG}$ is polynomial by the assumption. This is a contradiction that B3DM has an $\alpha$ approximation. 
	
	We now give a minor modification to adjust the result for BPS instead of BPB. The only necessary change is in \eqref{eq:insta}, where for BPS we define  \begin{equation}
		\label{eq:newIn}
		\mathcal{I}_i = (I_i,s,E \cup (T \times T))
	\end{equation} (instead of $ \mathcal{I}_i = (I_i,s,E)$ for BPB). That is, it holds that  $\cI_i$ as defined in \eqref{eq:newIn} is a BPS instance with a partition $T, I_i \setminus T$ into a clique $T$ and independent set $I_i \setminus T$ (The definition of $E$ is given in \eqref{eq:edges}). From here, the proof for BPS follows by symmetric argument to the proof for BPB. $\blacksquare$

\section{A $\frac{5}{3}$-approximation for \textnormal{BPB}} 
\label{sec:53}
 In this section, we continue to use maximization subroutines, this time for an absolute approximation algorithm for bin packing with bipartite conflicts. Note that Algorithm \textsf{ApproxBPC} (given in Section~\ref{sec:BP}) does not improve the absolute $\frac{7}{4}$-approximation for BPB of Epstein and Levin \cite{epstein2008bin}. Specifically, it fails to find a packing with better approximation ratio for instances with constant optimum. In these cases, the exact packing of relatively large items is crucial for the approximation guarantee; as an example, note that the absolute hardness of classic BP comes from instances with optimum $2$ \cite{garey1979computers}. 
 
  To improve Algorithm \textsf{ApproxBPC}, we rely on additional linear program for assigning items to partially packed bins; this slightly resembles Algorithm \textsf{MaxSolve}. Although, here we exploit the constant size of the optimum and that the graph is bipartite to obtain a nearly optimal assignment of items to the bins. 

For the following, fix a BPB instance $\ci = (I,s,E)$ and an error parameter $\eps = 0.0001$; let $T_{\ci} = \{v \in I~|~s(v) \leq \eps\}$ be the set of {\em tiny} items and let $B_{\ci} = I \setminus T_{\ci}$ be the set of {\em big} items. Such a classification of the items is useful as an optimal packing of the big items can be easily found using enumeration for fixed size optimum. We use a linear program to add the tiny items from one side of the conflict graph to a packing of the big items. We require that assigned items form an independent set with the already packed big items; however, since we add tiny items only from one side of the graph, no constraint is needed to guarantee that the assigned items form an independent set. 
Formally, 

\begin{definition}
	\label{def:LP}
	Given a \textnormal{BPB} instance $\ci = (I,s,E)$, a packing $\cA = (A_1 \ldots,A_t)$ of the big items of $I$, and $W \in \{X_I \cap T_{\ci},Y_I \cap T_{\ci}\}$ the tiny items from one of the sides of the bipartition of $G_{\ci}$. For all $i \in [t]$, let $Q_i = \{v \in W~|~A_i \cup \{v\} \in \textnormal{\textsf{IS}}(G_{\ci})\}$. Then, define the {\em assignment} of $\ci,\cA,W$ as the following linear program. 
	
	\begin{equation}
		\label{eq:LP53}
		\begin{aligned}
			\textnormal{\textsf{assignment}}(\ci,\cA,W):~~~~~~&	\max \sum_{i \in [t]} \bar{x}_{i,v}\\
			&	\textnormal{s.t.}\\
			&	 \bar{x}_{i,v} = 0 ~~~~~~~~~~~~~~~~~~~~~~~~~~~~~~~\forall i \in [t], v \in W \setminus Q_i\\
			&	 s(A_i)+\sum_{v \in W} \bar{x}_{i,v} \cdot s(v) \leq 1 ~~~~~~~~~\forall i \in [t]\\
			&	 \sum_{i \in [t]} \bar{x}_{i,v} \leq 1 ~~~~~~~~~~~~~~~~~~~~~~~~~~\forall v \in W\\
			&	 \bar{x}_{i,v} \in [0,1] ~~~~~~~~~~~~~~~~~~~~~~~~~~~\forall i \in [t], v \in W\\
		\end{aligned}
	\end{equation}
\end{definition}

For a basic feasible solution $\bar{x}$ of $\textnormal{\textsf{assignment}}(\ci,\cA,W)$, we use $\textsf{fractional}(\bar{x}) = \{v \in W~|~\exists i \in [t] \text{ s.t. } \bar{\lambda}_{i,v} \in (0,1)\}$ to denote the set of all items assigned fractionally, and by $\textsf{integral}(\bar{x}) = W \setminus \textsf{fractional}(\bar{x})$. As we do not use constraints for most edges in the conflict graph, in the next result we have some nice integrality properties of \eqref{eq:LP53}, based on the ratio between the number of variables to the number of independent constraints in \eqref{eq:LP53}.  The analysis of the next lemma uses the results of \cite{lenstra1990approximation,shmoys1993approximation}. 

\begin{lemma}
	\label{lem:integralityLP}
	There is a polynomial time algorithm \textnormal{\textsf{Round}} that given a \textnormal{BPC} instance $\ci = (I,s,E)$ such that $\OPT(\ci) \leq 100$, a packing $\cA = (A_1 \ldots,A_t)$ of $\ci \cap B_{\ci}$, and $W \in \{X_I \cap T_{\ci},Y_i \cap T_{\ci}\}$, returns a packing $\cB$ of $\ci \cap (B_{\ci} \cup W)$ such that $\#\cB = \#\cA$ and $|\textnormal{\textsf{items}}(\cB) \setminus B_{\ci}| \geq \OPT(	\textnormal{\textsf{assignment}}(\ci,\cA,W))-t$. 
\end{lemma}

By Lemma~\ref{lem:integralityLP}, if the packing of the big items has a small number $t$ of bins, and $\eps$ is small enough w.r.t. $t$, then we can assign all fractional items to one extra bin. Also, the tiny items from the opposite side of $W$ in the bipartition are packed in extra bins as well. This gives us the \textsf{Assign} subroutine, which enumerates over packings of the big items and for each packing performs the above assignment approach. We limit the enumeration only for packings with at most $100$ bins, so the running time remains polynomial. The pseudocode is given in Algorithm~\ref{alg:Assign}.

\begin{algorithm}[h]
	\caption{$\textsf{Assign}(\ci = (I,s,E), W)$}
	\label{alg:Assign}
		\begin{algorithmic}[1]
		\State{Initialize $\cB \leftarrow \textnormal{\textsf{\Generic}}(\ci)$.\label{step:AS:init}}
	\For{\textnormal{all packings $\cA$ of $\ci \cap B_{\ci}$ such that $\#\cB \leq 100$\label{step:AS:for}}}
			\State{Compute $ \cC \leftarrow \textnormal{\textsf{Round}}(\ci,\cA,W)$.\label{step:AS:compute}}
			\State{Let $\cB_{\cA} \leftarrow \cC \oplus \textnormal{\textsf{\Generic}}(\cI \setminus ( \textnormal{\textsf{items}}(\cC) \cup B_{\ci}))$.\label{step:AS:concate}}
		\If{$\#\cB> \#\cB_{\cA}$\label{step:AS:if}}
				\State{$\cB \leftarrow \cB_{\cA}$.\label{step:AS:AB}}
			\EndIf
	\EndFor
		\State{Return $\cB$.\label{step:AS:return}}
			\end{algorithmic}
\end{algorithm}

In the analysis of the next lemma, we focus on an iteration in which the big items are packed as in some optimal solution (there is such an iteration for $\OPT(\ci) \leq 100$). In this iteration, by Lemma~\ref{lem:integralityLP}, all items from $W$ can be packed by $\textnormal{\textsf{Round}}(\ci,\cA,W)$ by adding at most one bin. Other extra bins are added for the tiny items not in $W$ by \textsf{FFD}. 
\begin{lemma}
	\label{lem:assign}
	Given a \textnormal{BPP} instance $\ci = (I,s,E)$ and $W \in \{X_I \cap T_{\ci},Y_i \cap T_{\ci}\}$, Algorithm~\ref{alg:Assign} returns a packing $\cB$ of $\ci$, such that if $\OPT(\ci) \leq 100$ it holds that $\#\cB \leq \OPT(\ci)+1+\#\textnormal{\textsf{FFD}}(\ci \cap (T_{\ci} \setminus W))$. 
\end{lemma}

	We use the algorithm of \cite{epstein2008bin} for BPB as a subroutine; this handles well instances with optimum bounded by $3$. Specifically,   \begin{lemma}
	\label{lem:7/4}
	There is a polynomial time algorithm \textnormal{\textsf{App-BPB}} that is a $\frac{7}{4}$-approximation for \textnormal{BPB}. 
\end{lemma} Using the above, we have all subroutines for the algorithm. The algorithm computes $ \textsf{\Generic}$, $ \textsf{App-BPB}$ on the instance; in addition, the algorithm computes $\textsf{Assign}$  on both sides of the bipartition. Finally, the algorithm returns the best packing resulting from the above four attempts. The pseudocode is given in Algorithm~\ref{alg:abs}.

\begin{algorithm}[h]
	\caption{$\textsf{Abs-BPB}(\ci)$}
	\label{alg:abs}
		\begin{algorithmic}[1]
		\State{$\cA_1 \leftarrow \textsf{\Generic}(\ci)$, $\cA_2 \leftarrow \textsf{App-BPB}(\ci)$, $\cA_3 \leftarrow \textsf{Assign}(\ci,X_I \cap T_{\ci})$, $\cA_4 \leftarrow \textsf{Assign}(\ci,Y_I \cap T_{\ci})$.\label{step:absComp}}
		\State{Return $\argmin_{\cA \in \{\cA_1, \cA_2, \cA_3, \cA_4\}} \#\cA$.\label{step:absRet}}
		\end{algorithmic}
\end{algorithm}

The next theorem follows by a case analysis for different sizes of the optimum. For $\OPT(\ci) \leq 3$ we use Lemma~\ref{lem:7/4}. For $3<\OPT(\ci) \leq 100$, if there are many tiny items from both sides of the bipartition, the preferred packing is by Algorithm \textsf{\Generic}; otherwise, Algorithm \textsf{Assign} for the side of the bipartition dominating the number of tiny items gives the improve bound. Finally, for $\OPT(\ci)>100$ an absolute $\frac{5}{3}$-approximation can be easily obtained by the asymptotic guarantee of Algorithm \textsf{\Generic}. 
\begin{theorem}
	\label{lem:5/3}
	Algorithm~\ref{alg:abs} is a $\frac{5}{3}$-approximation for \textnormal{BPB}. 
\end{theorem}

\subsection{Deferred Proofs from Section~\ref{sec:53}}
\label{sec:proofs53}

	{\bf Proof of Lemma~\ref{lem:integralityLP}:}
	By the results of \cite{lenstra1990approximation,shmoys1993approximation} and also demonstrated in \cite{chekuri2005polynomial}, there is a polynomial time algorithm that given a basic feasible solution $\bar{x}$ for \eqref{eq:LP53}, finds a non feasible solution $\bar{y}$ for \eqref{eq:LP53} such that the following holds. \begin{enumerate}
		\item If $\bar{x}_{i,v} = 0$ , then $\bar{y}_{i,v} = 0$, and if $\bar{x}_{i,v} = 1$, then $\bar{y}_{i,v} = 1$.\label{cond:y}
		
		\item For each $i \in [t]$ there is at most one item $u_i \in W$ such that $\bar{y}_{i,u_i} \in (0,1)$ and  $s(A_i)+\sum_{u \in W \setminus \{u_i\}} \bar{x}_{i,v} \cdot s(v) \leq 1$.\label{cond:y0}
	\end{enumerate} Using $\bar{y}$, we define the following packing. For all $i \in [t]$ define $B_i = A_i \cup \{v \in W~|~\bar{y}_{i,v} = 1\}$. Observe that $\cB = (B_1,\ldots, B_t)$ is indeed a packing by \eqref{eq:LP53}, Condition~\ref{cond:y}, and Condition~\ref{cond:y0}. Finally, we note that $|\textnormal{\textsf{items}}(\cB) \setminus B_{\ci}| \geq \OPT(	\textnormal{\textsf{assignment}}(\ci,\cA,W))-t$ by Condition~\ref{cond:y} and Condition~\ref{cond:y0}. $\blacksquare$

\hfill \break

\noindent	{\bf Proof of Lemma~\ref{lem:assign}:}
	We use two auxiliary claims. \begin{myclaim}
		\label{claim:Assign:running}
		The running time of algorithm~\ref{alg:ptas} on $\ci,\eps$ is ${\textnormal{poly}(\ci)}$. 
	\end{myclaim}
	\begin{proof}
		By Lemma~\ref{lem:integralityLP} and Lemma~\ref{lem:Generic} the running time of each iteration of  the {\bf for} loop in Step~\ref{step:AS:for} is polynomial. Thus, to conclude that the overall running time is polynomail we bound the number of iterations of the loop. If $s(B_{\ci})> 100$, then there are no packings of $\ci \cap B_{\ci}$ with at most $100$ bins; thus, in this case the running time is polynomial by Step~\ref{step:AS:for} . Otherwise, by the definition of big items it holds that $|B_{\ci}| \leq \frac{100}{\eps}$. Thus, in this case the number of packings of $\ci \cap B_{\ci}$ with at most $100$ bins is bounded by $100^{{\frac{100}{\eps}}}$; since $\eps = 0.0001$, the number of such packins is a constant. Therefore,  the running time is polynomial by Step~\ref{step:AS:for}. 
	\qed \end{proof}

	\begin{myclaim}
		\label{claim:Assign:2}
		If $\OPT(\ci) \leq 100$, $\cB$ is a packing for $\ci$ such that $\#\cB \leq \OPT(\ci)+1+\#\textnormal{\textsf{FFD}}(\ci \cap (T_{\ci} \setminus W))$
	\end{myclaim}
	
	\begin{proof}
		Let $\OPT = (A_1, \ldots, A_n)$ be an optimal packing for $\ci$. Because $n = \OPT(\ci) \leq 100$, then there is an iteration of the {\bf for} loop of Step~\ref{step:AS:for} such that the considered packing $\cA = \cA(\OPT)$ of $\cI \cap B_{\ci}$ satisfies: 
		
		\begin{equation}
			\label{eq:A}
			\cA(\OPT) = (A_1 \cap B_{\ci}, \ldots, A_n \cap B_{\ci}).
		\end{equation} 
	
	Let $\cC(\OPT) = \textnormal{\textsf{Round}}(\ci,\cA(\OPT),W)$. We use the following inequality.
	
	 \begin{equation}
			\label{eq:s(W)}
			\begin{aligned}
				s\left(W \setminus (\textsf{items}(\cC(\OPT)) \cup B_{\ci})\right) \leq{} & \eps \cdot \left| W \setminus (\textsf{items}(\cC) \cup B_{\ci})  \right| \\
				\leq{} & \eps \cdot (|W|-(|W|-\OPT(\ci))) \\
				 \leq{} & \eps \cdot 100 \\
				 \leq{} & 1.
			\end{aligned}	
		\end{equation}

		The first inequality holds because $I \setminus (\textsf{items}(\cC) \cup B_{\ci}) \subseteq T_{\ci}$. The second inequality holds since by \eqref{eq:LP53} it holds that the optimum of $\textnormal{\textsf{assignment}}(\ci,\cA(\OPT),W)$ is $|W|$. Therefore, by Lemma~\ref{lem:integralityLP} it holds that $ \textnormal{\textsf{Round}}(\ci,\cA,W)$ is a packing in $\#\cA(\OPT) \leq \OPT(\ci)$ bins of $B_{\ci}$ and also at least $|W|-\OPT(\ci)$ items from $W$. The third inequality holds because $\OPT(\ci) \leq 100$. The last inequality holds since $\eps < 0.01$. 
		
		\begin{equation}
			\label{eq:assign:last}
			\begin{aligned}
				\#\cB \leq{} & \cB_{\cA(\OPT)}\\
				={} & \# \cC(\OPT) \oplus \textnormal{\textsf{\Generic}}(\cI \setminus ( \textnormal{\textsf{items}}(\cC(\OPT)) \cup B_{\ci}))\\
				={} & \# \cA(\OPT) \oplus \textnormal{\textsf{\Generic}}(\cI \setminus ( \textnormal{\textsf{items}}(\cC(\OPT)) \cup B_{\ci}))\\
				\leq{} & \OPT(\ci) + \#\textnormal{\textsf{FFD}}\left(  (\cI \setminus ( \textnormal{\textsf{items}}(\cC(\OPT)) \cup B_{\ci}))) \cap W \right) \\ +{} &  \#\textnormal{\textsf{FFD}}\left(  (\cI \setminus ( \textnormal{\textsf{items}}(\cC(\OPT)) \cup B_{\ci}))) \cap (T_{\ci} \setminus W) \right)\\
				\leq{} & \OPT(\ci)+1+\#\textnormal{\textsf{FFD}}(\ci \cap (T_{\ci} \setminus W)).
			\end{aligned}
		\end{equation} 
	
	The first inequality holds by Step~\ref{step:AS:if}. The first equality holds by Step~\ref{step:AS:compute} and Step~\ref{step:AS:concate}. The second equality holds by Lemma~\ref{lem:integralityLP}. The second inequality holds by \eqref{eq:A}, Step~\ref{step:GenericA}, and Step~\ref{step:GenericB} of Algorithm~\ref{alg:generic}. The last inequality holds because $	s\left(W \setminus (\textsf{items}(\cC(\OPT)) \cup B_{\ci})\right) \leq 1$ by \eqref{eq:s(W)}, thus the algorithm assign the items in $W \setminus (\textsf{items}(\cC(\OPT)) \cup B_{\ci})$ to a single bin. 
\qed	\end{proof}

	Observe that Algorithm~\ref{alg:Assign} returns a packing of $\ci$ also if $\OPT(\ci) \geq 100$ by Step~\ref{step:AS:init}, Step~\ref{step:AS:compute}, Step~\ref{step:AS:concate}, and Step~\ref{step:AS:AB}. Therefore, the proof of Lemma~\ref{lem:assign} follows by Claim~\ref{claim:Assign:running}, Claim~\ref{claim:Assign:2}. $\blacksquare$

		{\bf Proof of Lemma~\ref{lem:5/3}:} Let $\ci = (I,s,E)$ be a BPB instance. First, observe that the algorithm returns a packing for $\ci$ in polynomial time by Lemma~\ref{lem:7/4}, Lemma~\ref{lem:Generic}, Lemma~\ref{lem:assign}, Step~\ref{step:absComp}, and Step~\ref{step:absComp} of Algorithm~\ref{alg:abs}. For the approximation guarantee, we split the proof into several complementary cases by the optimum value of $\ci$. Let $\cB = \textsf{Abs-BPB}(\ci)$.

	\begin{myclaim}
		\label{claim:OPT=3}
		If $\OPT(\ci) \leq 3$ then Algorithm~\ref{alg:abs} returns a packing of at most $\frac{5}{3} \cdot \OPT(\ci)$ bins. 
	\end{myclaim}
	\begin{proof}
		By Lemma~\ref{lem:7/4}, it holds that  $\#\textsf{App-BPB}(\ci) \leq \frac{7}{4} \cdot \OPT(\ci)$. Observe that in this case, for $\OPT(\ci) = 1$ it holds that $\#\textsf{App-BPB}(\ci) = 1$; for $\OPT(\ci) = 2$ it holds that $\#\textsf{App-BPB}(\ci) = 3$; and for $\OPT(\ci) = 3$ it holds that $\#\textsf{App-BPB}(\ci) = 5$. Therefore, the proof follows by Step~\ref{step:absComp} and Step~\ref{step:absRet} of Algorithm~\ref{alg:abs}. \qed
	\end{proof}

	\begin{myclaim}
		\label{claim:OPT=4aux}
		If $\OPT(\ci) = 4$, $s(Y_{I} \cap T_{\ci}) > 1$, and  $s(X_{I} \cap T_{\ci}) > 1$, then $ \#\cB \leq \frac{5}{3} \cdot \OPT(\ci)$. 
	\end{myclaim}
	\begin{proof} We use the following inequality
		\begin{equation}
			\label{eq:sB}
			s(B_{\ci}) \leq s(I)-s(T_{\ci}) \leq \OPT(\ci)-s(Y_{I} \cap T_{\ci})-s(X_{I} \cap T_{\ci}) < 4-1-1 = 2. 
		\end{equation} The second inequality holds since $s(I) \leq \OPT(\ci)$.
		Assume towards a contradiction that $\#\textsf{\Generic}(\ci) \geq 7$. Therefore, by Step~\ref{step:GenericA}, Step~\ref{step:GenericB}, and Step~\ref{step:GenericReturn} of Algorithm~\ref{alg:generic} it holds that $\#\textnormal{\textsf{FFD}}(\ci \cap X_I) \geq 4$ or $\#\textnormal{\textsf{FFD}}(\ci \cap Y_I) \geq 4$. We split the proof into several complementary cases. \begin{enumerate}
			\item $\#\textnormal{\textsf{FFD}}(\ci \cap X_I) \geq 5$. Then, by Claim~\ref{claim:4} it holds that $s(X_I) \geq 4-4 \cdot \eps > 3$ (recall that $\eps<0.1$); this is a contradiction that $s(X_{I} \cap T_{\ci}) > 1$.\label{case:one}
			
			\item $\#\textnormal{\textsf{FFD}}(\ci \cap Y_I) \geq 5$. Symmetric argument as in Case~\ref{case:one}.

			\item $\#\textnormal{\textsf{FFD}}(\ci \cap X_I) = 4$.\label{case:two} Then, \begin{equation}
				\label{eq:101}
				s(Y_I) \leq s(I)-s(X_I) \leq \OPT(\ci)- (3-3 \cdot \eps) \leq 4-2.99 = 1.01.
			\end{equation} The second inequality holds by Claim~\ref{claim:4}. By \eqref{eq:101} and Observation~\ref{lem:FFD:one} it holds that $\#\textnormal{\textsf{FFD}}(\ci \cap Y_I) \leq 2$.
			Hence, we reach a contradiction, since by Step~\ref{step:GenericA}, Step~\ref{step:GenericB}, and Step~\ref{step:GenericReturn} of Algorithm~\ref{alg:generic} the returned packing is of at most $\#\textnormal{\textsf{FFD}}(\ci \cap X_I)+\#\textnormal{\textsf{FFD}}(\ci \cap Y_I) = 4+2 = 6$ bins. 
			
			\item $\#\textnormal{\textsf{FFD}}(\ci \cap Y_I) \geq 5$. Symmetric argument as in Case~\ref{case:one}.
		\end{enumerate} 
		
		Therefore, we reach a contradiction that $\#\textsf{\Generic}(\ci) \geq 7$. Thus, by Step~\ref{step:absComp} and Step~\ref{step:absRet} it holds that $\#\cB \leq 6$. \qed
	\end{proof}
	\begin{myclaim}
		\label{claim:OPT=4}
		If $\OPT(\ci) = 4$, then Algorithm~\ref{alg:abs} returns a packing of at most $\frac{5}{3} \cdot \OPT(\ci)$ bins. 
	\end{myclaim}
	\begin{proof}
		We split the proof into several sub-cases. \begin{enumerate}
			\item $s(X_{I} \cap T_{\ci}) \leq 1$. Then, $$\#\cB \leq \OPT(\ci)+1+\#\textnormal{\textsf{FFD}}(\ci \cap (T_{\ci} \setminus W)) \leq \OPT(\ci)+2 = 6.$$ The first inequality holds by Lemma~\ref{lem:assign}, Step~\ref{step:absComp}, and Step~\ref{step:absRet} of Algorithm~\ref{alg:abs}. The second inequality holds by Step~\ref{step:FFDif} and Step~\ref{step:FFDnotnew} of Algorithm $\textnormal{\textsf{FFD}}$ and that $s(X_{I} \cap T_{\ci}) \leq 1$ (i.e., no second bin is opened in the course of \textsf{FFD} as all items fit in a single bin). The third inequality holds since $\OPT(\ci) = 4$.\label{case:1}
			
			\item $s(Y_{I} \cap T_{\ci}) \leq 1$. Then, $\#\cB \leq  6$ by symmetric arguments to Case~\ref{case:1}.  
			
			\item $s(Y_{I} \cap T_{\ci}) > 1$ and  $s(X_{I} \cap T_{\ci}) > 1$. The proof follows by Claim~\ref{claim:OPT=4aux}. \qed
		\end{enumerate}  
	\end{proof}

	We define $t = \floor{\frac{2 \cdot \OPT(\ci)}{3}}$ as the extra number of bins allowed to reach our approximation guarantee. 	\begin{myclaim}
		\label{claim:OPT=5aux}
		If $5 \leq \OPT(\ci) \leq 100$, $\#\textnormal{\textsf{FFD}}(X_{I} \cap T_{\ci}) > t-1$, and  $\#\textnormal{\textsf{FFD}}(Y_{I} \cap T_{\ci}) > t-1$,	then $(\#\textnormal{\textsf{FFD}}(X_{I})-1) \cdot (1-\eps) \leq s(X_{I})$ and $(\#\textnormal{\textsf{FFD}}(Y_{I})-1) \cdot (1-\eps) \leq s(Y_{I})$.
	\end{myclaim}
	\begin{proof}
		We use the following inequality. \begin{equation}
			\label{eq:sB2}
			\begin{aligned}
				s(B_{\ci}) \leq{} & s(I)-s(T_{\ci}) \\
				\leq{} & \OPT(\ci)-s(Y_{I} \cap T_{\ci})-s(X_{I} \cap T_{\ci}) \\ 
				\leq{} & \OPT(\ci)-2 \cdot \left(t-1\right) \cdot (1-\eps) \\ 
				\leq{} & \OPT(\ci)-2 \cdot t+2+2\eps \cdot (t-1)\\ 
				\leq{} & \OPT(\ci)-2 \cdot \left(\frac{2}{3} \cdot \OPT(\ci)-\frac{2}{3}\right)+2+200 \cdot \eps\\ 
				\leq{} &3+\frac{1}{3}+0.1-\frac{\OPT(\ci)}{3}\\ 
				\leq{} &3+\frac{2}{3} - \frac{\OPT(\ci)}{3}\\ 
				\leq{} &2\\ 
			\end{aligned}
		\end{equation} The third inequality holds because $\#\textnormal{\textsf{FFD}}(X_{I} \cap T_{\ci}) > t-1$,  $\#\textnormal{\textsf{FFD}}(Y_{I} \cap T_{\ci}) > t-1$, and by Claim~\ref{claim:4}. The fifth inequality holds since $\floor{\frac{2 \cdot \OPT(\ci)}{3}} \geq \frac{2 \cdot \OPT(\ci)}{3}-\frac{2}{3}$. The sixth inequality holds since $\eps < 0.0001$. The last inequality holds since $\OPT(\ci) \geq 5$. Because $\#\textnormal{\textsf{FFD}}(X_{I} \cap T_{\ci}) > t-1$ and  $\#\textnormal{\textsf{FFD}}(Y_{I} \cap T_{\ci}) > t-1$, then for $\OPT(\ci) \geq 5$ it holds that $s(T_{\ci} \cap X_I) > 1$ and $s(T_{\ci} \cap Y_I) > 1$ ($\textsf{FFD}$ returns a packing with a single bin for instances with total size at most one). Thus, by \eqref{eq:sB2} and that $s(T_{\ci} \cap X_I) > 1$ and $s(T_{\ci} \cap Y_I) > 1$, by Claim~\ref{claim:4} the proof follows. 
	\qed \end{proof}
	
	\begin{myclaim}
		\label{claim:OPT=5}
		If $5 \leq  \OPT(\ci) \leq 100$, then Algorithm~\ref{alg:abs} returns a packing of at most $\frac{5}{3} \cdot \OPT(\ci)$ bins. 
	\end{myclaim}
	\begin{proof}
		The proof follows by the several sub-cases below. \begin{enumerate}
			\item $\#\textnormal{\textsf{FFD}}(X_{I} \cap T_{\ci}) \leq t-1$. Then, $$\#\cB \leq \OPT(\ci)+1+\#\textnormal{\textsf{FFD}}(\ci \cap (T_{\ci} \setminus W)) \leq \OPT(\ci)+1+t-1 = \OPT(\ci)+t \leq \frac{5}{3} \cdot \OPT(\ci).$$ The first inequality holds by Lemma~\ref{lem:assign}, Step~\ref{step:absComp}, and Step~\ref{step:absRet} of Algorithm~\ref{alg:abs}. The second inequality holds by the definition of Algorithm $\textnormal{\textsf{FFD}}$ and that $\#\textnormal{\textsf{FFD}}(X_{I} \cap T_{\ci}) \leq t-1$ (i.e., no second bin is opened in the course of FFD).\label{case:11}
			
			\item $s(Y_{I} \cap T_{\ci}) \leq 1$. Then, $\#\cB \leq \frac{5}{3} \cdot \OPT(\ci)$ by symmetric arguments to Case~\ref{case:11}.  
			
			\item $\#\textnormal{\textsf{FFD}}(X_{I} \cap T_{\ci}) > t-1$ and  $\#\textnormal{\textsf{FFD}}(Y_{I} \cap T_{\ci}) > t-1$.  Therefore,

			\begin{equation*}
				\begin{aligned}
					\#\cB \leq{} & \textsf{\Generic}(\ci) \\
					\leq{} & \#\textnormal{\textsf{FFD}}(X_{I}) \oplus \#\textnormal{\textsf{FFD}}(Y_{I})\\
					\leq{} & \frac{s(X_I)}{1-\eps}+1+\frac{s(Y_I)}{1-\eps}+1\\
					\leq{} & (s(X_I)+s(Y_I)) \cdot (1+2 \cdot \eps)+2\\
					\leq{} & s(I) \cdot 1.0002+2\\
					\leq{} & \OPT(\ci) +0.0002 \cdot \OPT(\ci)+2\\
					\leq{} & \OPT(\ci) +3\\
					\leq{} & \frac{5}{3} \cdot \OPT(\ci).\\
				\end{aligned}
			\end{equation*} The first inequality holds by Step~\ref{step:absRet} of Algorithm~\ref{alg:abs}. The second inequality holds by Step~\ref{step:GenericA}, Step~\ref{step:GenericB}, and Step~\ref{step:GenericReturn} of Algorithm~\ref{alg:generic}. The third inequality holds by Claim~\ref{claim:OPT=5aux}. The sixth inequality holds since $\eps = 0.0001$. The last inequalities hold since $5 \leq \OPT(\ci) \leq 100$. 
			
		\end{enumerate}  
	\qed \end{proof}

	In the following proof, we use the next result of \cite{epstein2008bin}.

	\begin{myclaim}
		\label{claim:epsh2}
		For all \textnormal{BPP} instance $\ci = (I,s,E)$, it holds that $\OPT(\ci \cap X_I)+\OPT(\ci \cap Y_I) \leq \frac{3}{2} \cdot \OPT(\ci)+1$. 
	\end{myclaim} 
	Now, 
	\begin{myclaim}
		\label{claim:OPT=6}
		If $ \OPT(\ci) \geq 100$, Algorithm~\ref{alg:abs} returns a packing of at most $\frac{5}{3} \cdot \OPT(\ci)$ bins. 
	\end{myclaim}
	\begin{proof}

		\begin{equation*}
			\begin{aligned}
				\#\cB \leq{} & \textsf{\Generic}(\ci) \\
				\leq{} & \#\textnormal{\textsf{AsymptoticBP}}(X_{I}) \oplus \#\textnormal{\textsf{AsymptoticBP}}(Y_{I})\\
				\leq{} & 1.02 \cdot \OPT(\ci \cap X_I)+1.02 \cdot \OPT(\ci \cap Y_I)\\
				={} &1.02 \cdot \left(  \OPT(\ci \cap X_I)+ \OPT(\ci \cap Y_I)\right)\\
				\leq{} &1.02 \cdot (1.5  \cdot \OPT(\ci) +1)\\
				\leq{} & 1.53  \cdot \OPT(\ci) +1.02\\
				\leq{} & 1.53  \cdot \OPT(\ci) +0.02 \cdot \OPT(\ci)\\
				\leq{} & 1.55  \cdot \OPT(\ci)\\
				\leq{} & \frac{5}{3} \cdot \OPT(\ci).\\
			\end{aligned}
		\end{equation*} The first inequality holds by Step~\ref{step:absComp} and Step~\ref{step:absRet} of Algorithm~\ref{alg:abs}. The second inequality holds by Step~\ref{step:GenericA}, Step~\ref{step:GenericB}, and Step~\ref{step:GenericReturn} of Algorithm~\ref{alg:generic}. The third inequality holds by Lemma~\ref{lem:rothvos}. The fourth inequality holds by Claim~\ref{claim:epsh2}. The sixth inequality uses $ 100 \leq \OPT(\ci)$.

	\qed \end{proof}
	Finally, the proof of Theorem~\ref{lem:5/3} follows by Claim~\ref{claim:OPT=3}, Claim~\ref{claim:OPT=4}, Claim~\ref{claim:OPT=5}, and Claim~\ref{claim:OPT=6}. $\blacksquare$
\mysection{Complete Multi-partite Conflict Graphs}
\label{sec:partite}

In this section we give an asymptotic lower bound for {\em Bin Packing with multi-partite conflict graph (BPM)}, and show that it matches the (absolute) approximation guarantee of Algorithm \textsf{\Generic}.  
The main idea is to use the $\frac{3}{2}$-hardness of approximation result for classic bin packing. Specifically, we {\em duplicate} a BP instance $\ci$ with $\OPT(\ci) = 2$ to generate a BPM instance. Interestingly, 
we derive an asymptotic hardness result for the BPM instance using the absolute hardness of BP. The next results resolve the complexity status of BPM.   
\begin{lemma}
	\label{bpm:1}
	For any $\alpha<\frac{3}{2}$, there is no asymptotic $\alpha$-approximation for \textnormal{BPM}, unless \textnormal{P=NP}. 
\end{lemma} 
\begin{lemma}
	\label{bpm:2}
	Algorithm~\ref{alg:generic} is a $\frac{3}{2}$-approximation for \textnormal{BPM}. 
\end{lemma}

We now prove Lemmas~\ref{bpm:1} and~\ref{bpm:2}.  Recall that a multipartite graph $G = (V,E)$ has a partition $V_1, \ldots, V_k$ of $V$ such that for all $i \in [k]$ it holds that $V_i$ is an independent set and for all $u \in V_i,v \in V_j, i,j \in [k], i \neq j$ it holds that $(u,v) \in E$. Therefore, for simplicity we use $\mathcal{I} = \big(I, \{G_1, \ldots,G_n\}, s \big)$ to denote a BPM instance, where $I$ is a set of items, $\{G_1, \ldots,G_n\}$ is the unique partition of $I$ into independent sets, and $s$ is the size function. Also, let $\ci_{j} = (G_i,s)$ be the corresponding BP instance to the $j$-th set $G_j$. Finally, given a packing $\cA$ w.r.t. $\ci$ let $\cA(G_j)$ be the packing w.r.t. $\ci_j$ containing all bins in $\cA$ that contain items from $G_j$ (observe that each bin in $\cA$ contains item from only one $G_j$ by the definition of complete multipartite graph). 
\begin{lemma}
	\label{lem:BP}
	For any \textnormal{BPM} instance $\mathcal{I} = \big(I, \{G_1, \ldots,G_n\}, s \big)$ it holds that $$\OPT(\mathcal{I}) = \sum_{j \in [n]} \OPT(\ci_j).$$
\end{lemma}

\begin{proof}
	Let $\cA = (A_1, \ldots, A_m)$ be an optimal packing for $\mathcal{I}$. Since for each $j \in [n]$ the items of $G_j$ must be packed in bins that contain only items from $G_j$, there must be at least $\OPT(\ci_j)$ such bins in the packing $\cA$. If there are more than $\OPT(\ci_j)$ such bins it follows that $\cA$ is not optimal for $\mathcal{I}$. $\blacksquare$ 
\end{proof}

\begin{lemma}
	\label{lem:asymptotic}
	For any \textnormal{BPM} instance $\mathcal{I} = \big(I, \{G_1, \ldots,G_n\}, s \big)$ such that for all $j \in [n]$ it holds that $\OPT(\ci_j) = 2$, and a packing $\cA = (A_1, \ldots, A_m)$ of $\mathcal{I}$ with $\#\cA < 1.5  \cdot \OPT(\mathcal{I})$, there is $j \in [n]$ such that $|\cA(G_j)| \leq 2$ and $\cA(G_j)$ is a packing for $G_j$. 
\end{lemma}

\begin{proof}
	Assume towards a contradiction that for all $j \in [n]$ it holds that $\cA(G_j) \geq 3$. Therefore, 
	
	\begin{equation}
		\label{2}
		m = \sum_{j \in [n]} \cA(G_j) \geq 3n = \frac{3}{2} \sum_{j \in [n]} OPT(G_j) = \frac{3}{2} OPT(\mathcal{I}).
	\end{equation}
	
	The first equality is because in BPM each bin contains items only from one group. The first inequality is by the assumption that $\cA(G_j) \geq 3$. The second equality is because for all $j \in [n]$ it holds that $OPT(G_j) = 2$. The last equality is by Lemma~\ref{lem:BP}. By \eqref{2} it holds that $m \geq 1.5 \cdot \OPT(\mathcal{I})$ in contradiction.  $\blacksquare$\end{proof}

We define a {\em duplication} of a BP instance, which is simply duplicate the set of items and giving each item the original size. Formally, 

\begin{definition}
	\label{def:1}
	Let $G=(I,s)$ be a BP instance. A {\em duplication of $G$} is a BP instance $G'=(I',s')$ where there is a bijection $\sigma:I' \rightarrow I$ such that for all $i \in I'$ it holds that $s(i) = s\left(\sigma(i)\right)$.
\end{definition}

\begin{lemma}
	\label{lem:dup}
	Let $G=(I,s)$ be a BP instance and $G'=(I',s')$ be a duplication of $G$. Then, there is a polynomial time algorithm that given a packing of $G'$ with size $m$ finds a packing of $G$ with size $m$. 
\end{lemma}

\begin{proof}
	Let $(A'_1,\ldots, A'_m)$ be a packing of $G'$ and let $\sigma:I' \rightarrow I$ be the bijection promised by Definition~\ref{def:1}. Observe that $\sigma$ can be found in polynomial time by sorting the items according to their sizes. For all $a \in [m]$ define $A_a = \{\sigma(i) ~|~ i \in A'_a\}$. Note that for all $i \in I$ there is $a \in [m]$ such that $i \in A_a$ since there is $i' \in I'$ such that $\sigma(i') = i$ because $\sigma$ is a bijection. In addition, for all $a \in [m]$ it holds that $s(A_a) = s(A'_a) \leq 1$. The first equality is by the definition of $\sigma$.  Hence, it follows that $(A_1, \ldots, A_m)$ is a packing of $G$.  $\blacksquare$
\end{proof}

\begin{lemma}
	\label{lem:NP}
	For any $n \in \mathbb{N}$, it is \textnormal{NP}-hard to find a packing with size $2$ of a \textnormal{BP} instance $G = (I,s)$ with $|I| \geq n$ and $\OPT(G)=2$.
\end{lemma}

\begin{proof}
	We use a reduction from BP, which is known to be NP-hard even for instances $G$ with $OPT(G) = 2$ \cite{garey1979computers}. The reduction is defined as follows given a BP instance $G = (I,s)$ with $OPT(G) = 2$. If $|I| \geq n$, then return $G' = G$. Otherwise, let $G' = (I \cup I',s')$ where $|I'| = n-|I|$, for all $i \in I$ it holds that $s'(i) = s(i)$, and for all $i \in I'$ it holds that $s'(i) = 0$. Given a packing of $G'$ with size $m$ we can find a packing to $G$ with size $m$, by removing the items from $I'$. Moreover, given a packing of $G$ with size $m$ we can find a packing to $G'$ with size $m$ by adding each $i \in I'$ to one of the bins in the packing arbitrarily. This results in a feasible packing since the size of the items in $I'$ is $0$. The claim follows.    $\blacksquare$
\end{proof}

\noindent{\bf Proof of Lemma~\ref{bpm:1}:} Assume towards a contradiction that there is $\alpha<1.5$ such that there is an asymptotic $\alpha$-approximation for BPM, and denote by $X$ the algorithm achieving the above approximation. Denote by $X(\mathcal{I})$ the packing returned by $X$ on some instance $\mathcal{I}$ and let $c = \frac{1.5-\alpha}{2}$. By the definition of asymptotic $\alpha$-approximation, there is $n \in \mathbb{N}$ such that for all $n' \geq n$ and a BPM instance $\mathcal{I}$ with $\OPT(\mathcal{I}) = n'$ it holds that 
	
	\begin{equation}
		\label{1}
		\#X(\mathcal{I}) \leq \alpha\cdot \OPT(\mathcal{I})+c\cdot \OPT(\mathcal{I})= (\alpha+c)\cdot \OPT(\mathcal{I}) < 1.5 \cdot \OPT(\mathcal{I}). 
	\end{equation}
	
	The first inequality is because $X$ is an asymptotic $\alpha$-approximation for BPM. The second inequality is because $\alpha<1.5$. Let $\cj = (I,s)$ be a BP instance such that $|I| \geq n$. Let $(G_1,s), \ldots,(G_{n},s)$ be $n$ distinct duplications of $\ci$ and let $\mathcal{I}_n = \big(I_{n}, \{G_1, \ldots,G_{n}\}, s \big)$ 
	where $I_{n} = \bigcup_{j \in [n]} G_j$. Note that the construction of $\mathcal{I}_n$ can be done in polynomial time because $|I| \geq n$. Let $\cA = X(\mathcal{I}_n)$ and let $\cA = (A_1, \ldots, A_m)$. By Lemma~\ref{lem:BP} it follows that $\OPT(\mathcal{I}_n) \geq n$ and by Equation~\ref{1} it holds that $|A|<1.5 \cdot \OPT(\mathcal{I})$. Therefore, by Lemma~\ref{lem:asymptotic} there is $j \in [n]$ such that $|\cA(G_j)| \leq 2$ and that $\cA(G_j)$ is a packing for $\ci_j$.  Then, using $\cA(G_j)$, we can find in polynomial time a packing for $\cj$ with size at most $2$ by Lemma~\ref{lem:dup}. Unless P=NP, this is a contradiction by Lemma~\ref{lem:NP}.  $\blacksquare$

~\\
\noindent{\bf Proof of Lemma~\ref{bpm:2}:}
	Let $\mathcal{I} = (I, \{G_1, \ldots,G_n\}, s )$ be a BPM instance. Then,
	
	$$\#\textsf{\Generic}(\ci) = \sum_{j \in [n]} \#\textsf{FFD}(\ci_j) \leq \sum_{j \in [n]} \frac{3}{2} \OPT(\ci_j)= \frac{3}{2} \sum_{j \in [n]}  \OPT(\ci_j) = \frac{3}{2} \OPT(\mathcal{I}).$$
	
	The first equality is by Algorithm~\ref{alg:generic}. The first inequality is is since \textsf{FFD} is a $\frac{3}{2}$-approximation for BP. The last equality is by Lemma~\ref{lem:BP}.  $\blacksquare$
\tableofcontents
\end{document}